\documentclass[10pt, article, twocolumn]{IEEEtran}

\usepackage{graphicx}
\usepackage{float}
\usepackage{amsfonts}
\usepackage{amsmath}
\usepackage{amssymb}
\usepackage{cite}
\usepackage{bm}
\usepackage{mathrsfs}
\usepackage{epsfig}
\usepackage{algorithm}
\usepackage{algorithmic}
\usepackage{amsthm}
\usepackage{color}
\usepackage{subfigure}

\newtheorem{proposition}{Proposition}

\begin{document}
\author{\authorblockA{Thakshila   Wimalajeewa\thanks{$^{\dagger}$ Dept. of EECS, Syracuse University, Syracuse, NY 13244, USA, Email: twwewelw@syr.edu, varshney@syr.edu}$^{\dagger}$ ~\IEEEmembership{Member,~IEEE},    Pramod   K. Varshney$^{\dagger}$~\IEEEmembership{Fellow,~IEEE} and Wei   Su\thanks{$^{\ddagger}$  Email: weisu888@yahoo.com}$^{\ddagger}$~\IEEEmembership{Fellow,~IEEE}}
}
\title{Detection of Single vs Multiple Antenna Transmission Systems  Using   Pilot Data}

\maketitle

\begin{abstract}
In this paper, we consider the problem of classifying  the transmission system  when it is not known \emph{a priori} whether the transmission is via a single antenna or multiple antennas. The receiver  is assumed to be  employed with a known number of antennas.    In  a data frame transmitted by most multiple input multiple output  (MIMO) systems, some pilot or training  data is inserted for symbol timing synchronization and  estimation of   the channel. Our goal is to perform MIMO transmit antenna  classification using this  pilot data.  More specifically, the problem of determining the  transmission  system  is  cast as a multiple hypothesis testing problem where the number of hypotheses is equal to the maximum number of transmit antennas.  Under the assumption of  receiver having the exact knowledge of pilot data used for timing synchronization and channel estimation, we consider maximum likelihood (ML)  and correlation   test statistics to classify the  MIMO transmit system. When only probabilistic knowledge of pilot data is available at the receiver, a hybrid-maximum likelihood (HML) based test statistic is constructed using the expectation-maximization  (EM) algorithm. The performance of the proposed algorithms is illustrated via  simulations and comparative  merits  of different  techniques in terms of the computational complexity and  performance  are  discussed.
\end{abstract}

\footnotetext[1]{Distribution Statement A: Approved for public release; distribution  is unlimited}

%
\IEEEpeerreviewmaketitle

\section{ Introduction}
The rapidly growing interest in software defined and cognitive radios, and pervasive ubiquitous sensing has been fuelled by the development of intelligent communication systems    \cite{Mitola_93,mitola_pc99,haykin1}. In spite of the tremendous advances made, practical and efficient realization of these systems for envisaged applications in different domains  requires that several obstacles  be overcome.  Signal recognition with minimal amount of prior information is one of the key requirements in most intelligent communication systems  \cite{su_iet07,Ramkumar_cir09,Choqueuse_TSP10,Dobre_2015,Xu_VT2010}.
With the increasingly diverse radio frequency signal population ranging from simple narrow-band analog and digital modulations to wideband digital modulation schemes utilizing multiple transmit antennas, the problem of blind signal recognition continues to be extremely challenging.

Multiple-input-multiple-output (MIMO) systems  is a physical layer technology that can provide many benefits through multiple antennas and advanced signal processing.   With the availability of perfect timing and frequency information, and when the parameters related to the observation model, such as,  channel state information (CSI), the number of transmit  antennas and encoding schemes  are available at the receiver, the problem of recovering the transmitted signal has been  studied widely \cite{Larsson_SPM09}. However, when the receiver and the transmitter are noncooperative or have limited cooperation, MIMO signal recognition  is challenging \cite{Eldemerdash_CoMS16}. More specifically, when the transmitters are employed with sophisticated communication technologies, there is a need for new blind signal recognition algorithms  that are able to operate in such environments. There  are recent efforts that consider how to determine  the modulation schemes used by MIMO transmit antennas  when the transmitter and receiver are equipped with multiple antennas   \cite{Choqueuse_MTA1,Hassan_twc12,Muhlhaus_vtc1,
Muhlhaus_CL13,Zhu_MILCOM1,Marey_SPL14,Kharbech_WCL14,Chikha_WCL14}. However, in most of these  works, it was assumed that the number of transmit antennas is  known at the receiver.

In a typical MIMO system, multiple antennas at the transmitter  require multiple RF chains which consist of amplifiers, analog to digital converters, mixers, etc., that are typically very expensive. An approach for reducing the cost while maintaining a high potential data rate of a MIMO system is to employ a reduced number of RF chains at the transmitter and attempt to optimally allocate each chain to one of the larger number of transmit antennas which are usually cheaper  \cite{Jin_tsp13}. Thus, when the transmitter changes the number of transmit antennas depending on the application or the requirement, the receiver finds it is challenging  to determine  the transmission scheme in order to perform a given  inference task.

 The problem of determining the number of sources/signals  in  array signal processing  has been addressed in early works  \cite{Wax_85,AOUADA_2003,Xu_tsp94,Chen_tsp91}. Determining the number of transmit antennas  used in a MIMO system  has been addressed by several authors in \cite{Oularbi_2013,Oularbi_CL13,Somekh_MILCOM07,Ohlmer_VTC08,Ohlmer_VTC08_2,Mohammadkarimi_2015,Marey_TWC12}. In \cite{Somekh_MILCOM07}, the authors consider the problem of determining  the number of transmit antennas utilizing  objective information
theoretic criteria such as Akaike's information criterion (AIC) and minimum descriptor length (MDL). Another approach to estimate the number of transmit  antennas in a MIMO orthogonal frequency division multiplexing  (OFDM) system based on pilot patterns is presented in \cite{Oularbi_2013,Oularbi_CL13}. In these works, the problem of  determining  the number of transmit antennas, $n$,  is formulated as that of detecting $n$ pilot patterns in an OFDM system.  In \cite{Ohlmer_VTC08,Ohlmer_VTC08_2}, the authors consider the determination   of the number of transmit antennas in a MIMO-OFDM system  when each burst is preceded by a preamble. Their algorithm is based on  the estimation of  channels with different hypothetical numbers
of transmit antennas. In \cite{Mohammadkarimi_2015},  blind  detection of
the number of transmit antennas is studied  by exploiting the time-diversity of
 fading channels. Classification of multiple antenna systems in the presence of possible transmission impairments exploiting  cyclostationarity property of space-time block codes (STBCs)  was considered in \cite{Marey_TWC12}.  In all these works, it is assumed that the symbol timing synchronization at the receiver has been  achieved.     In \cite{Berenguer_TSP05,Jin_tsp13}, the problem of antenna selection at the transmitter is considered in which the goal is to decide which set of antennas are  to be used for transmission  based on different performance criteria at the receiver.

Our goal in this paper is to  study the problem of determining  the number of transmit antennas  prior to performing symbol timing synchronization and channel estimation when the transmitter and the receiver  have limited cooperation.   In a practical  MIMO system, frequency and symbol timing synchronization, and channel estimation need to  be performed  at the receiver before recovering  data. While there are several techniques proposed for symbol timing and channel estimation  \cite{Wang_WCNC07}, the use of pilot symbols  in MIMO systems has been studied extensively for symbol timing synchronization \cite{Naguib_JSC98,Bliss_TSP10} and  channel estimation \cite{Budianu_ICC01,Biguesh_TSP06,Naguib_JSC98,Shin_TC08,Hassibi_TIT03}.    The knowledge of the pilot sequences assigned to  each transmit antenna may not be available exactly at the receiver  depending  on the level  of cooperation between the transmitter and the receiver. In this paper, we study the problem of MIMO transmit  system classification  based on pilot data. In particular, we propose several test statistics, develop  algorithms  and  discuss relative  merits of different  algorithms considering asynchronous (in the absence of symbol timing synchronization) and synchronous cases (in the presence of symbol timing synchronization).  In the asynchronous case,  maximum likelihood (ML) and correlation classifiers  are considered assuming that an exact knowledge of the pilot sequence assignment for each antenna is available   at the receiver.  The results are then specified  for  the synchronous case.  When the exact knowledge of the pilot data assignment of each transmit antenna is not known at the receiver, we develop hybrid maximum likelihood (HML) based classification schemes  using the expectation-maximization (EM)
algorithm for both asynchronous and synchronous cases. We further specify the results when  the transmission system is binary; i.e., transmission can be performed either with a single antenna or multiple antennas with a known number of antennas. Then, the MIMO transmit system classification problem reduces to the problem of  MIMO vs. single-input multiple-output (SIMO) detection.  The performance of each classifier/detector  is illustrated  via simulations  and  comparative  merits  of different  techniques are  discussed.


The rest of the paper is organized as follows. In Section \ref{sec_background}, the background and the problem formulation are presented. In Section \ref{sec_knownPilot}, the problem of asynchronous MIMO transmit antenna   classification is discussed when the receiver has the exact knowledge of the pilot sequences. The results are provided for   the   synchronous case as well for  the case when the transmit system is binary.  In Section \ref{sec_UnknownPilot}, the analyses are extended to the case  where  the  exact knowledge of the pilot sequences  is  not known at the receiver. Simulation results are given in Section \ref{simulation} and  concluding remarks  are given  in Section \ref{conclusion}.

\subsection*{Notation and Terminology}
We use 'Tx' and 'Rx' to denote 'transmit' and   'receive', respectively. Lower case letters, e.g.,  $x$,  are used to denote scalars and functions  while boldface lower case letters, e.g.,  $\mathbf x$,  are used to denote vectors. Boldface upper case letters, e.g., $\mathbf X$,  are used to denote matrices. The $(j,k)$-th element of a matrix $\mathbf X$ is denoted by $(\mathbf X)_{jk}$.  Matrix transpose and Hermitian transpose operators are denoted by  $(\cdot)^T$ and $(\cdot)^H$, respectively. The notation $||\cdot||_F$ is used to denote the Frobenius norm. The trace operator is denoted by $\mathrm{tr}(\cdot)$.

\section{Background and Problem Formulation}\label{sec_background}
Consider a MIMO communication system with $m$ receiver (Rx) antennas and $n$ transmit (Tx) antennas (often called $m\times n$ MIMO)  as depicted  in Fig. \ref{fig_MIMO_model}.
\begin{figure}[h!]
\centerline{\epsfig{figure=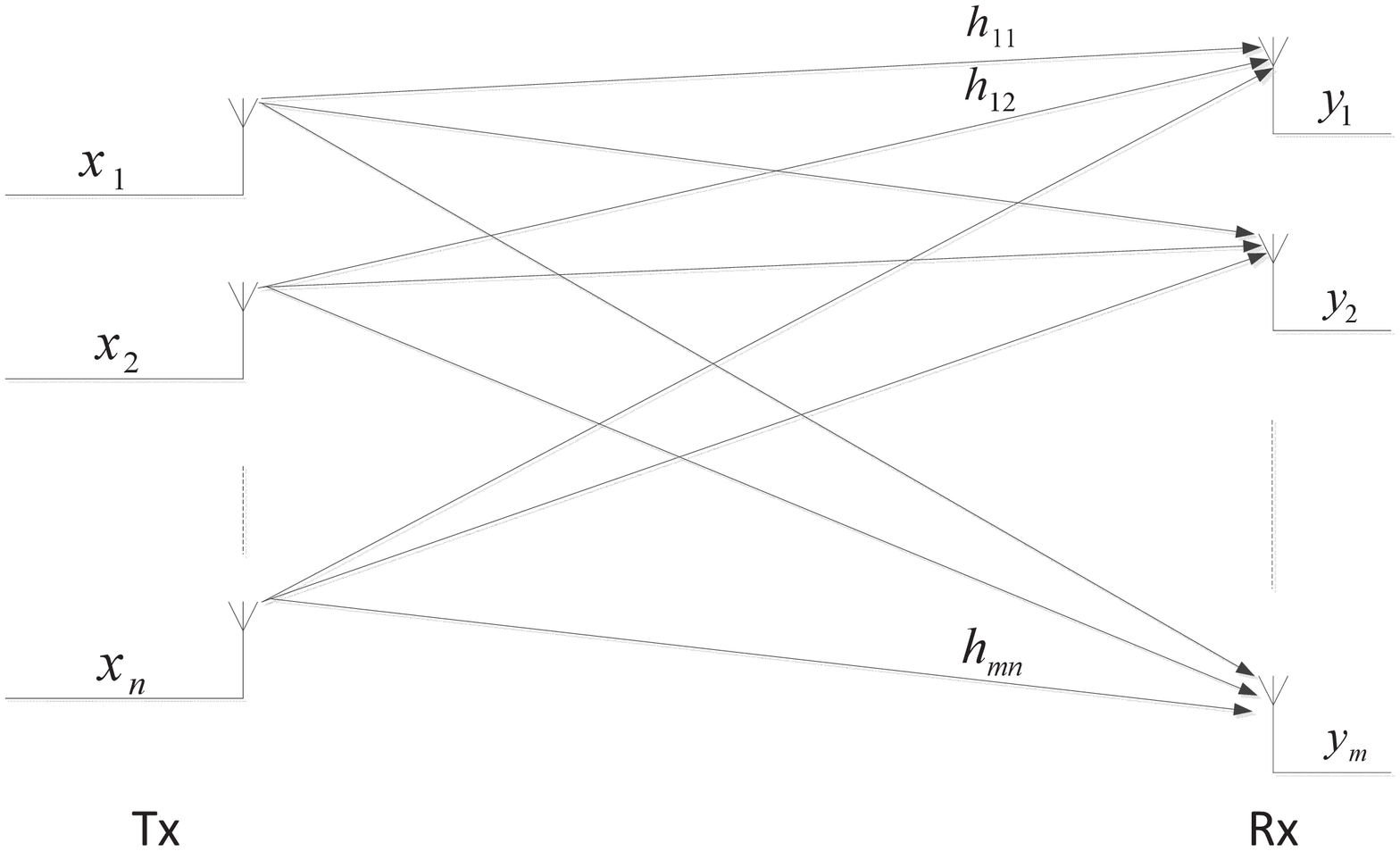,width=8.50cm}}
\caption{MIMO system with $n$ Tx antennas and $m$ Rx antennas}\label{fig_MIMO_model}
\end{figure}
The symbols   transmitted by each Tx antenna are assumed to undergo independent fading.

Let $x_i[l]$ be the $l$-th symbol transmitted by the $i$-th antenna. Note that $x_i[l]$   can be either a pilot/training symbol  or (encoded) data  symbol. After pulse shaping, the transmitted signal by the $i$-th Tx antenna can be   expressed as,
\begin{eqnarray*}
r_i(t) = \beta \underset{l}{\sum} x_i[l] g(t-lT_s)
\end{eqnarray*}
where  $T_s$ is the symbol period,  and  $g(.)$ denotes  the transmitted pulse. We denote the symbol rate as $R_s = \frac{1}{T_s}$.  Further, in general, we assume different Tx powers  for data symbols and pilot symbols. Thus,  $\beta = \beta_d$ if $x_i[l]$ represents a data symbol and  $\beta = \beta_p$ if $x_i[l]$ represents a pilot symbol. We
assume that the transmit power is equally distributed across the transmitted antennas, thus, the amplitudes are identical across them.   The transmitted pulse is assumed to be a square root raised-cosine ($\sqrt{RC}$) pulse \cite{Proakis:book}  which is given by,
\begin{eqnarray*}
g(t) = \frac{4\epsilon}{\pi\sqrt{T_s}} \frac{\cos((1+\epsilon)\pi t /T_s) + \frac{\sin((1-\epsilon)\pi t /T_s)}{(4\epsilon t /T_s)}}{1 - (4 \epsilon t /T_s)^2}
\end{eqnarray*}
where $\epsilon$ is the roll-off factor.

With a given delay $\tau_q$, the  received signal at the $q$-th  Rx antenna is given by \cite{Proakis:book},
\begin{eqnarray}
y_q(t) = \sum_{i=1}^n h_{qi} \underset{l}{\sum} \beta  x_i[l] g(t-lT_s - \tau_q ) + w_q(t)\label{Rx_signal}
\end{eqnarray}
for $0< t< T_0$ where $h_{qi}$ is the complex fading channel coefficient from the $i$-th Tx antenna to the $q$-th Rx antenna (assumed to be fixed for given $T_0$),  and $w_q(t)$ is the additive  noise. For a given number of Tx antennas, we further assume that $\tau_q = \tau$ for all $q$ which is a reasonable assumption when the Rx antennas are located not very far from each other.
If  perfect timing and frequency information is available  at the receiver, the $l$-th received   symbol in (\ref{Rx_signal}) after matched filtering can be written as
\begin{eqnarray}
y_q[l] = \sum_{i=1}^n h_{qi}  \beta  x_i[l] + w_q[l]\label{Rx_signal_sync}
\end{eqnarray}
for $l=1,2,\cdots$ and $q=1,\cdots,m$.

However, when  the time delay is unknown, it has to be estimated before performing matched filtering. Further, even if the symbol timing is available or symbols are synchronized, it is required to estimate the channel state information at  the receiver before recovering the data  symbols. In order to perform all these operations, the number of Tx antennas used by the transmitter should be known at the receiver.

In this paper, we  classify   the Tx system in terms of the number of transmit antennas   before recovering or extracting required information about  the transmitted signal. When the maximum number of Tx antennas is known, this problem can be formulated as a multiple hypothesis testing (or classification) problem.
Classification is performed based on training/pilot data used for symbol timing synchronization and channel estimation.  In many communication systems, timing synchronization and channel estimation  are  achieved with the aid of pilot data \cite{Budianu_ICC01,Biguesh_TSP06,Naguib_JSC98,Shin_TC08,Hassibi_TIT03,Bliss_TSP10}.   For example, in Fig. \ref{fig_MIMO_dataformat}, one general format of the data transmitted by each Tx  antenna is illustrated in a time division multiple access (TDMA) framework \cite{TDMA}. With this timing and framing structure,   each Tx antenna transmits a burst of length $L_b$ which contains pilot data for symbol timing, pilot data for channel learning and encoded (informative) data.   This particular format (with/without  slight modifications) with the same pilot sequence repeating periodically at a given antenna is used in commercial systems such as  narrowband TDMA/STCM-based modems  \cite{Tarokh_TIT1}.


 \begin{figure*}
\centerline{\epsfig{figure=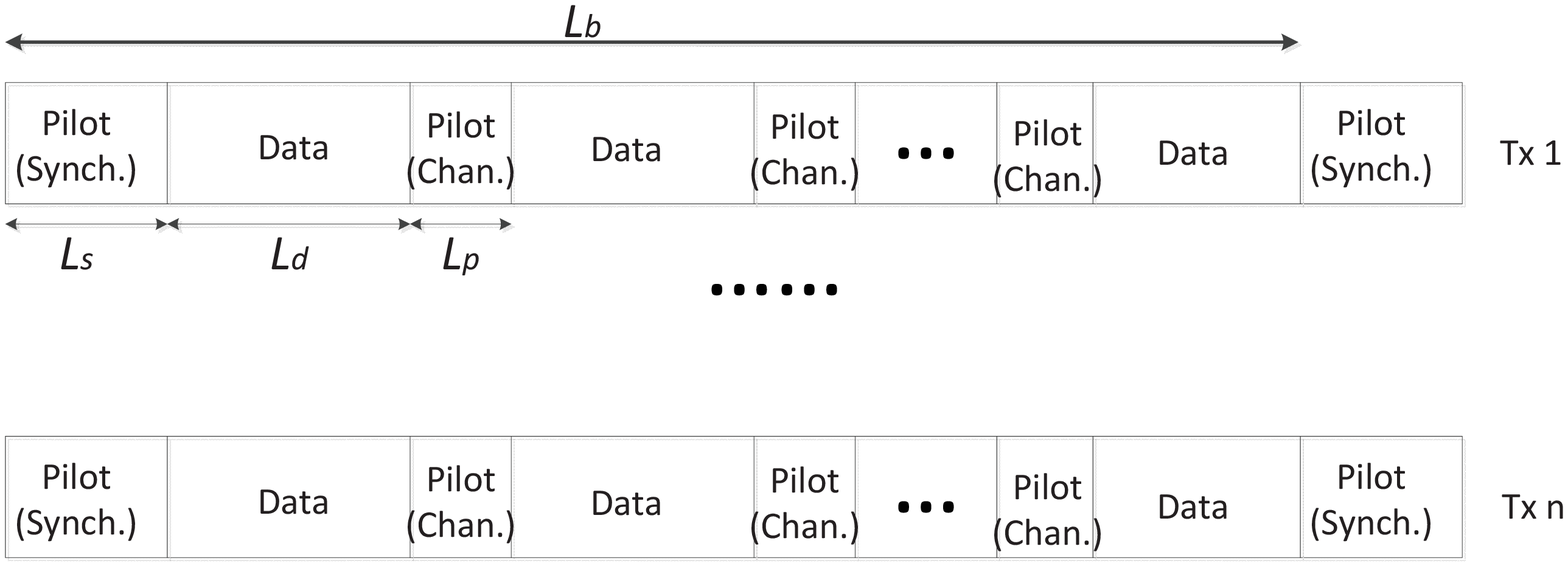,width=10.50cm}}
\caption{Format of a data burst sent by multiple antennas with TDMA \cite{Tarokh_TIT1}}\label{fig_MIMO_dataformat}
\end{figure*}

We consider several scenarios. In the asynchronous case,  we use the pilot data used for timing synchronization to classify the Tx system. On the other hand,  when perfect timing synchronization is  achieved at the receiver (without the use of pilot symbols) or when  time delay is  negligible, we use pilot sequences used for channel estimation to classify the Tx system. Depending on which pilot data is used and the knowledge available  at the receiver about the pilot data,  the decision statistics used for Tx system  classification are different.

\section{Tx System  Classification with Exact Knowledge of Pilot Data at the Receiver}\label{sec_knownPilot}
Each Tx antenna is assumed to use an equal number of training/pilot symbols (say $L_s$)  for timing synchronization   in  each  transmitted frame of length $L_b$. Channel estimation is performed periodically within a  burst of length $L_b$ with a different set of  pilot data of length $L_p$. We assume that the pilot sequences  are selected from a given pool of   orthogonal sequences. The use of orthogonal sequences for symbol timing and channel learning has been used in many MIMO systems  \cite{Tarokh_TIT1}.
 Although it is reasonable to assume to have  different number of pilot symbols per antenna \cite{Samardzija_TSP03}, here we consider    that  the pilot symbols  are of the  same length. In this section, we assume  that the  pilot sequences $x_i[l]$'s  for $l=1,\cdots,L_s ( \mathrm{or} ~L_p)$ assigned to the $i$-th  antenna for $i=1,\cdots,n$ are  known at the receiver. This is a valid assumption when the Tx and Rx have  limited cooperation.

Let  the vector representation of $y_q(t)$, $w_q(t)$ and  $ r_i(t)$ be $ [y_{q}[1], \cdots, y_{q}[L_s]]^T \equiv \mathbf y_q$, $[w_{q}[1], \cdots, w_{q}[L_s]]^T \equiv \mathbf w_q$ and $[r_{i}[1], \cdots, r_{i}[L_s]]^T \equiv \mathbf r_i$, respectively for $q=1,\cdots,m$ and $i=1,\cdots,n$ and $L_s= \frac{T_0}{T_s}$ where $T_0$ is the observation time period for pilot symbols. Then the received signals at the $m$ receivers can be written in matrix form as,
  \begin{eqnarray}
\mathbf Y= \mathbf H \mathbf R_{\tau} + \mathbf W\label{matrix_timing}
\end{eqnarray}
where $\mathbf Y = [\mathbf y_1, \cdots, \mathbf y_m]^T$, $\mathbf R_{\tau}$ represents  the equivalent discrete time representation of  $r_i(t - \tau)$ for $i=1,\cdots,n$ similar to   $\mathbf R = [\mathbf r_1, \cdots, \mathbf r_n]^T$, and $\mathbf W = [\mathbf w_1, \cdots, \mathbf w_m]^T$. The noise process  is  assumed to be uncorrelated and  Gaussian (and the elements of $\mathbf W$ are assumed to be iid Gaussian with mean zero and variance $\sigma_w^2$). The matrix $\mathbf H$ is the channel matrix where   $(\mathbf H)_{qi} = h_{qi}$  for $q=1,\cdots, m$ and $i=1,\cdots,n$.

After matched filtering  with a given delay $\tau$,  (\ref{matrix_timing}) is equivalently represented by,
\begin{eqnarray*}
\mathbf Y_{\tau}= \mathbf H \mathbf R + \mathbf W
\end{eqnarray*}
where the $(q,k)$-th element of $\mathbf Y_{\tau}$ is given by,
\begin{eqnarray*}
(\mathbf Y_{\tau})_{qk} = y_q^{\tau}[k] = \int_{T_0} y_q(t) g(t-kT_s -\tau) dt.
\end{eqnarray*}

Let there be a maximum of $n_{\max}$  antennas  possible for the MIMO system. The transmission system can choose  any number of antennas, $n$,  from $[1, n_{\max}]$ with equal probability.    Then, the MIMO Tx system classification problem can be treated as a multiple hypothesis testing problem with  $n_{\max}$  number of hypotheses. Each hypothesis $\mathcal H_j$,  for $j=1,\cdots, n_{\max}$,  represents a MIMO system with a given number of antennas (including SIMO system when $n=1$).
The multiple hypothesis testing problem is given by,
\begin{eqnarray}
\mathcal H_j: \mathbf Y_{\tau_j} = \mathbf H_j \mathbf R_j + \mathbf W\label{mult_hyp_timing}
\end{eqnarray}
for $j=1, \cdots, n_{\max}$.
 It is noted that, (\ref{mult_hyp_timing}) has to be solved in the presence of  unknown  parameters $\mathbf H_j$,  and $\tau$ under $\mathcal H_j$ for $j=1,\cdots,n_{\max}$.

In the  case where timing information is available at the receiver (i.e. the synchronous case), the received signal matrix can be represented by,
  \begin{eqnarray}
  \mathbf Y= \mathbf H \mathbf R + \mathbf W\label{obs_matrix_sync}
  \end{eqnarray}
where now  $(\mathbf Y)_{qk}=y_q[k] =\sum_{i=1}^n  \beta_p h_{qi}   x_i[k] + w_q[k]$ and $(\mathbf R)_{ik} = r_i[k] = \beta_p x_i[k] $ for $q=1,\cdots,m$ and $k=1,\cdots, L_p$. Then, similar to (\ref{mult_hyp_timing}), the synchronous MIMO classification problem   can be cast as,
\begin{eqnarray}
\mathcal H_j: \mathbf Y = \mathbf H_j \mathbf R_j + \mathbf W. \label{mult_hyp_timing_synch}
\end{eqnarray}
Thus, in the synchronous case, MIMO  classification is required to be performed based on (\ref{mult_hyp_timing_synch}) in the presence of  unknown   $\mathbf H_j$ under   $\mathcal H_j$ for $j=1,\cdots, n_{\max}$.


\subsection{ML Classifier}
In the maximum likelihood (ML) framework, the  unknown parameters are found so that the likelihood function under each hypothesis is maximized.  For the asynchronous case, let $\hat{\tau}_j$ ($\tau$ under $\mathcal H_j$ is denoted by $\tau_j$ for clarity) and $\hat{\mathbf H}_j$ denote the ML estimates of $\tau$ and $\mathbf H_j$, respectively.
Under $\mathcal  H_j$, given $\mathbf H_j$ and $\tau_j$, $\mathbf Y$ has the following matrix variate normal distribution:
\begin{eqnarray}
&~&p(\mathbf Y_{\tau_j} |\mathbf H_j, \tau_j,  \mathcal H_j) \nonumber\\
&=& \frac{1}{(2\pi\sigma_w^2)^{mL_p/2}} e^{-\frac{1}{2} \mathrm{tr}\left[\frac{1}{\sigma_w^2} (\mathbf Y_{\tau_j} - \mathbf H_j \mathbf R_j)^H (\mathbf Y_{\tau_j} - \mathbf H_j \mathbf R_j)\right]}.\label{pdf_Hj}
\end{eqnarray}
The ML  estimates  (MLEs) of $\mathbf H_j$ and $\tau_j$   are computed  so that the probability density function (pdf) in (\ref{pdf_Hj}) is maximized with respect to $\mathbf H_j$ and $\tau_j$. They  can be obtained  as the solutions to the following two equations:
\begin{eqnarray}
\mathbf H_j = \mathbf Y_{\tau_j} \mathbf R^H_j(\mathbf R_j \mathbf R_j^H)^{-1}\label{eq1_MIMO}
\end{eqnarray}
and
\begin{eqnarray}
\underset{i,k}{\sum}(\mathbf H_j \mathbf R_j)_{ik} \frac{\partial y_{i}^{\tau_j}[k]}{\partial \tau_j} = 0\label{eq2_MIMO}
\end{eqnarray}
where  $y_i^{\tau_j}[k] = \int_{T_0} y_i(t) g(t-kT_s -\tau_j) dt$. It is noted that (\ref{eq1_MIMO}) and (\ref{eq2_MIMO}) are obtained by letting the partial derivatives of $p(\mathbf Y_{\tau_j} |\mathbf H_j, \tau_j,  \mathcal H_j)$ in (\ref{pdf_Hj}) with respect to $\mathbf H_j$ and $\tau_j$ to be zero.  From (\ref{eq1_MIMO}) and (\ref{eq2_MIMO}), $\hat{\tau}_j$ can be computed as the solution for $\tau_j$ to the following equation:
\begin{eqnarray}
\underset{i,k}{\sum}(\mathbf Y_{\tau_j}{\mathbf P}_j)_{ik} \frac{\partial y_{i}^{\tau_j}[k]}{\partial \tau_j} = 0 \label{eq2b_MIMO}
\end{eqnarray}
where
\begin{eqnarray}
{\mathbf P}_j = \mathbf R^H_j(\mathbf R_j \mathbf R_j^H)^{-1}{\mathbf R}_j.\label{P_j}
 \end{eqnarray}
 Once $\hat{\tau}_j$ is obtained, $\hat{\mathbf H}_j$ is given by,
\begin{eqnarray*}
\hat{\mathbf H}_j= \mathbf Y_{\hat{\tau}_j} \mathbf R^H_j(\mathbf R_j \mathbf R_j ^H)^{-1}
\end{eqnarray*}
which is the least squared (LS) channel estimator for the estimated time delay  $\hat{\tau}_j$.


Then, the ML classifier selects the hypothesis which gives the maximum likelihood function:
\begin{eqnarray*}
\hat j  &=& \underset{j}{\arg\max} \frac{1}{\sigma_w^2} \left\{\mathrm{tr}(\mathbf R_j^H \hat{\mathbf H}_j^H \mathbf Y_{\hat{\tau}_j})
-\frac{1}{2} \mathrm{tr}(\mathbf R_j^H \hat{\mathbf H}_j^H \hat{\mathbf H}_j \mathbf R_j) \right\}.
\end{eqnarray*}

For the synchronous case,   the  ML classifier  reduces to
\begin{eqnarray}
\hat j = \underset{j}{\arg \min}~ \mathrm{tr}((\mathbf P_j^{\bot}) \mathbf Y^H \mathbf Y)\label{ML_sync}
\end{eqnarray}
where  $\mathbf P_i$ is as defined in (\ref{P_j}) and $(\mathbf Y)_{qk} = \sum_{l=1}^n \beta_p h_{ql}   x_l[k] + w_q[k]$.

\subsubsection{MIMO vs. SIMO detection}
In this section, we specify the results when the Tx system is binary; i.e., transmission can be via either a single antenna or multiple antennas (i.e., SIMO vs. MIMO). This becomes a binary detection problem with:
\begin{eqnarray}
\mathcal H_1 (\mathrm{MIMO}): ~ \mathbf Y_{\tau} &=& \mathbf H_M \mathbf  R_{M} + \mathbf W\nonumber\\
\mathcal H_2 (\mathrm{SIMO}): ~\mathbf Y_{\tau} &=& \mathbf H_S \mathbf R_{S} + \mathbf W \label{hyp_MIMO_SIMO_timing}
\end{eqnarray}
where the subscripts $M$ and $S$ are  used to denote the multiple  antenna and single antenna cases, respectively. It is noted that $\mathbf H_S$ is a $m \times 1$  vector and  $\mathbf R_{S}$ is  a row vector.
In this case, the ML test statistic, which is known as the generalized likelihood ratio test (GLRT),  reduces to
\begin{eqnarray*}
\Lambda_{GLRT, async.} &=& \frac{1}{\sigma_w^2} \left\{\mathrm{tr}(\mathbf R_M^H \hat{\mathbf H}_M^H \mathbf Y_{\hat{\tau}_M} - \mathbf R_S^H \hat{\mathbf H}_S^H \mathbf Y_{\hat{\tau}_S}) \right.\nonumber\\
&-& \left.\frac{1}{2} \mathrm{tr}(\mathbf R_M^H \hat{\mathbf H}_M^H \hat{\mathbf H}_M \mathbf R_M - \mathbf R_S^H \hat{\mathbf H}_S^H \hat{\mathbf H}_S \mathbf R_S)\right\}.
\end{eqnarray*}

In the synchronous case, the GLRT test statistic can be expressed as
\begin{eqnarray}
\Lambda_{GLRT,sync} =  \mathrm{tr}((\mathbf P_S^{\bot} - \mathbf P_M^{\bot}) \mathbf Y^H \mathbf Y)\label{GLRT_Test}
\end{eqnarray}
where  $\mathbf P_S^{\bot} = \mathbf I - \mathbf P_S$ and $\mathbf P_M^{\bot} = \mathbf I - \mathbf P_M$ with
$
\mathbf P_S= \mathbf R_S^{H}(\mathbf R_S \mathbf R_S^{H})^{-1} \mathbf R_S
$
and
$
\mathbf P_M= \mathbf R_M^{H}(\mathbf R_M \mathbf R_M^{H})^{-1} \mathbf R_M$, as defined before.
It is noted that  the GLRT test statistic in the synchronous case can be computed fairly easily.

\subsubsection{Design of threshold for GLRT in the synchronous case}\label{threshold}
In the following,  we design the threshold of the synchronous  GLRT detector so that the probability of false alarms is kept under a desired value.  The probabilities  of false alarm and  detection of the GLRT are given by,
 \begin{eqnarray*}
 P_f = Pr (\Lambda_{GLRT,sync} \geq \tau_g |\mathcal H_2 )
 \end{eqnarray*}
 and
  \begin{eqnarray*}
 P_d = Pr (\Lambda_{GLRT,sync} \geq \tau_g |\mathcal H_1 )
 \end{eqnarray*}
respectively, where $\tau_g$ is the threshold. With Gaussian approximation, we can show that the threshold $\tau_g$  that  maintains  the probability of false  alarm at $\alpha$ can be approximated as given in the following proposition.

\begin{proposition}\label{prop1}
The threshold of the synchronous  GLRT detector which ensures $ P_f\leq \alpha$ can be approximated by
\begin{eqnarray}
\tau_g \approx  (\tilde\sigma Q^{-1}(\alpha) + \tilde\mu).\label{taug}
\end{eqnarray}
where \begin{eqnarray}
\tilde \mu &\approx&  (\tilde{\mathbf R}_S^{H} (\mathbf P_M - \mathbf P_S) \tilde{\mathbf R}_S)\sum_{l=1}^m |\hat{\mathbf H}_S(l)|_2^2  \nonumber\\
&+& m~\sigma_w^2 \mathrm{tr}(\mathbf P_M - \mathbf P_S)\label{tmu}
\end{eqnarray}
and
 \begin{eqnarray}
 \tilde\sigma^2 &\approx& 2 \sigma_w^2 ( \tilde{\mathbf R}_S^{H}( \mathbf P_M  + \mathbf P_S - 2\mathbf P_M \mathbf P_S) \tilde{\mathbf R}_S\sum_{l=1}^m{|\hat{\mathbf H}_S(l)|^2} \nonumber\\
 &+&  m \sigma_w^4\mathrm{tr}(\mathbf P_M+\mathbf P_S -2\mathbf P_M \mathbf P_S) \label{tsigma}
 \end{eqnarray}
 where $\sigma_w^2 = \mathbb E\{|(\mathbf W)_{ij}|^2\}$ and $\tilde{\mathbf R}_S = \mathbf R_S^T$.
\end{proposition}

\begin{proof}
See  Appendix A.
\end{proof}

The computation of the test statistic for the ML classifier (and the GLRT detector)  in a closed-form is difficult in the asynchronous case  since it is difficult to obtain  a closed-form solution  to   (\ref{eq2b_MIMO}).  However, the ML classifier   has a simple closed-form expression in the synchronous case. In both cases, (asynchronous/synchronous) the performance of the classifier   depends on the ML estimates of unknown parameters which depend on the length of the pilot sequences. In the   following,  we consider  statistics for a suboptimal detector,  known to be  correlation detector,   to perform MIMO system  classification  where estimation of all the unknown parameters is not required.

\subsection{Correlation Classifier}
For  a single-input single output (SISO) link with AWGN channels,  ML estimation of the delay involves the maximization of  the correlation between the received signal and the transmitted signal (however, with an unknown channel matrix, this optimality may  no longer holds). Motivated by this, we expect that for MIMO system as considered in this paper, the correlation  between the received signal matrix and the transmitted pilot symbols  will provide a reasonable  decision statistic for the classification  problem in (\ref{mult_hyp_timing}) in the presence of unknown delay.
In the following,  we consider correlation classifier for the multiple hypothesis testing problem which does not require the estimation of $\mathbf H_j$. For the asynchronous case, the correlation classifier is given by,
\begin{eqnarray}
\hat j = \underset{j}{\arg\max}~ \mathrm{tr}(\mathbf Y_{\hat{\tau}_j} (\mathbf R_{j}^{H} \mathbf R_{j}) \mathbf Y_{\hat{\tau}_j}^H) \label{corre_asynch}
\end{eqnarray}
where
\begin{eqnarray}
\hat{\tau}_j = \underset{\tau_j}{\arg\max} ||\mathbf C_{\tau_j}||_F^2\label{hat_tau_M}
\end{eqnarray}
and $\mathbf C_{\tau_j} =\mathbf Y_{\tau_j} \mathbf R_j^H $.
 It is noted that,   this decision statistic in  (\ref{corre_asynch}) is fairly easy to compute.
For the synchronous case, (\ref{corre_asynch}) reduces to,
\begin{eqnarray*}
\hat j = \underset{j}{\arg\max} ~\mathrm{tr}(\mathbf Y (\mathbf R_{j}^{H} \mathbf R_{j}) \mathbf Y^H).
\end{eqnarray*}

\section{Tx System  Classification When  the Exact Knowledge of Pilot Data is Not Available}\label{sec_UnknownPilot}
For both  ML and the correlation classifiers  as considered in Section \ref{sec_knownPilot} (in both asynchronous and synchronous cases), the exact  knowledge of the pilot data assigned to each Tx antenna is assumed to be available at the receiver. However, when  the cooperation between the transmitter and the receiver is too limited, this assumption may be too restrictive. Next, we consider  the case where the receiver has only probabilistic information regarding the pilot sequences assigned to each Tx antenna.  First,  we consider the asynchronous case.

Let each Tx antenna use a  pilot sequence $\mathbf x_i=[x_i[1], \cdots, x_i[L_s]]^T$  for symbol timing synchronization  drawn as a column of a given orthogonal  matrix $\mathbf  Q$ where $\mathbf Q$ is known to the receiver. However, the receiver is not aware of the exact column assigned to a given  antenna from $\mathbf Q$. Thus, the receiver assumes that any column of  $\mathbf Q$ is assigned to $\mathbf x_i$ randomly with the same probability for given $i$.

Let
\begin{eqnarray}
\mathbf X= \left[
\begin{array}{cc}
\mathbf x_1^T\\
 \cdot \\
 \mathbf x_n^T
 \end{array}\right]\label{X_unknown}
  \end{eqnarray}
be the  $n\times L_s$ matrix containing all the different  pilot sequences used by $n$ antennas. Then,  there can be $L_s^n = \frac{L_s!}{(L_s -n)!}$ possibilities for $\mathbf X$. The $k$-th realization of $\mathbf X$ is denoted by $\mathbf X_k$. The joint pdf of $\mathbf Y_{\tau}$ marginalized over $\mathbf X_k$ under $\mathcal H_j$ is given by,
\begin{eqnarray}
p(\mathbf Y_{\tau_j} | \mathbf H_j, \tau_j; \mathcal H_j) = \underset{\mathbf X_k}{\sum} p(\mathbf Y_{\tau_j} | \mathbf H_j, \tau_j, \mathbf X_k; \mathcal H_j) p(\mathbf X_k)\label{marg_pdf_MIMO}
\end{eqnarray}
where $p(\mathbf X_k)$ is the probability that $\mathbf X_k$ being selected. When $\mathbf X_k$ is selected uniformly, we have $p(\mathbf X_k) = \frac{1}{L_s^n}$. In order to perform likelihood ratio based classification, we consider the  hybrid maximum likelihood (HML) based approach, where the unknown parameters  are  estimated so that the marginalized pdf  under each hypothesis is maximized.    Since finding these estimators is computationally intractable due to marginalization, we provide a numerical technique based on the expectation-maximization (EM) algorithm.

\subsection{HML Based MIMO Antenna Classification via EM}\label{HML_EM_Async}
The EM algorithm is an iterative
numerical method  which can be used to compute ML estimates. It is well suited when  ML estimation
is intractable due to the presence of hidden (unobserved) data. The outline of the  EM algorithm is given in \cite{rubin_jrs77}, and the use of the EM algorithm in HML based classification is discussed in our previous work in  \cite{Wimalajeewa_amc_14} in a different application scenario.   For the problem addressed in this paper, the actual pilot sequences $\mathbf  X$ can be treated as hidden
data. Then, \emph{complete data} can be expressed as $[\mathbf Y, \mathbf X]$.  Starting from initial estimates $[\hat{\mathbf H}_j^{(0)}, \hat{\tau}_j^{(0)}]$, the  two operations as in (\ref{EM_twstep}) are performed at the  $r$-th iteration under $\mathcal H_j$.
\begin{figure*}
\begin{eqnarray}
\mathrm{\bf E-Step}:g(\mathbf H_j, \tau_j | \hat{\mathbf H}_j^{(r)}, \hat{\tau}_j^{(r)}) &=& \underset{\mathbf X} {\sum}\log p(\mathbf Y_{\tau_j} | \mathbf X, \mathbf H_j) p(\mathbf X |  \mathbf Y_{\hat{\tau}_j^{(r)}}, \hat{\mathbf H}_j^{(r)})\nonumber\\
\mathrm{\bf M-Step}: \{ \hat{\mathbf H}_j^{(r+1)}, \hat{\tau}_j^{r+1}\} &=& \underset{\mathbf H_j, \tau_j} {\arg\max} ~ g(\mathbf H_j, \tau_j | \hat{\mathbf H}_j^{(r)}, \hat{\tau}_j^{(r)}).\label{EM_twstep}
\end{eqnarray}
\end{figure*}

Let $\alpha_l^{(r)} = p(\mathbf X = \mathbf X_l | \mathbf Y_{\hat{\tau}_j^{(r)}},   \hat{\mathbf H}_j^{(r)})$ which can be expressed as,
\begin{eqnarray*}
\alpha_l^{(r)} = \frac{p(\mathbf Y_{\hat{\tau}_j^{(r)}} |  \hat{\mathbf H}_j^{(r)}, \mathbf X = \mathbf X_l)}{ \sum_{k=1}^{L_s^n} p(\mathbf Y_{\hat{\tau}_j^{(r)}} |  \hat{\mathbf H}_j^{(r)}, \mathbf X = \mathbf X_k) }
\end{eqnarray*}
where  $\left(\mathbf Y_{\hat{\tau}_j^{(r)}}\right)_{jk}=\int_{T_0} y_j(t) g(t-kT_s -\hat{\tau}_j^{(r)})$. Then,  the estimates  $\hat{\mathbf H}_j^{r+1}$ and $\hat{\tau}_j^{r+1}$ at the $(r+1)$-th iteration can be found as the solution for $\mathbf H_j$ and $\tau_j$ in the following two equations:
\begin{eqnarray*}
\mathbf H_j = \mathbf Y_{\tau_j}\left(\sum_{l=1}^{L_s^n} \alpha_l^{(r)} \beta_p\mathbf X_l^H\right)\left(\sum_{l=1}^{L_s^n} \alpha_l^{(r)} \beta_p^2\mathbf X_{l}\mathbf X_l^H\right)^{-1}
\end{eqnarray*}
and
\begin{eqnarray*}
 \sum_{l=1}^{L_s^n} \alpha_l^{(r)} \underset{i,k}{\sum}(\beta_p\mathbf H_{j} \mathbf X_l)_{i,k} \frac{\partial y_i^{\tau_j}[k]}{\partial \tau_j} = 0
\end{eqnarray*}
where $ y_i^{\tau_j}[k] = \int_{T_0} y_j(t) g(t-kT_s -\tau_j) dt$.
Once the unknown parameters  under both hypotheses are found via the EM algorithm, the HML based  asynchronous MIMO Tx system  classifier  is given by,
\begin{eqnarray}
\hat j =\underset{1\leq j\leq n_{\max}}{\arg\max}{p(\mathbf Y_{\hat{\tau}_j} | \hat{\mathbf H}_j, \hat{\tau}_j; \mathcal H_j)}\label{HML_classifier}
\end{eqnarray}
where $\hat{\tau}_j$ and  $ \hat{\mathbf H}_j$, are the estimates found by the EM algorithm for each $j$.

In the synchronous case,  the EM algorithm is implemented  to estimate  only  $\mathbf H_j$ under  $\mathcal H_j$. More specifically, under $\mathcal H_j$, the estimate of $\mathbf H_j$ at the $r$-th iteration is found such that,
\begin{eqnarray*}
\hat{\mathbf H}_j^{(r+1)} = \mathbf Y\left(\sum_{l=1}^{L_p^n} \alpha_l^{(r)}\beta_p \mathbf X_l^H\right)\left(\sum_{l=1}^{L_p^n} \alpha_l^{(r)}\beta_p^2 \mathbf X_{l}\mathbf X_l^H\right)^{-1}
\end{eqnarray*}
where
\begin{eqnarray*}
\alpha_l^{(r)} = \frac{p(\mathbf Y |  \hat{\mathbf H}_j^{(r)}, \mathbf X = \mathbf X_l)}{ \sum_{k=1}^{L_p^n} p(\mathbf Y |  \hat{\mathbf H}_j^{(r)}, \mathbf X = \mathbf X_k) }.
\end{eqnarray*}
Then, a test statistic is found similar to that in (\ref{HML_classifier}) letting $\tau_j=0$. Compared to performing HML  in the asynchronous case, it is noted that the computation of MLEs via EM algorithm for the synchronous case is fairly easy since at each iteration, the unknown parameters can be found in a closed-form.


\begin{figure*}
\centering
\subfigure[$4\times 2$ MIMO vs. SIMO, $\mathrm{SNR}=0 ~dB$]{%
\includegraphics[width=0.4\textwidth,height=!]{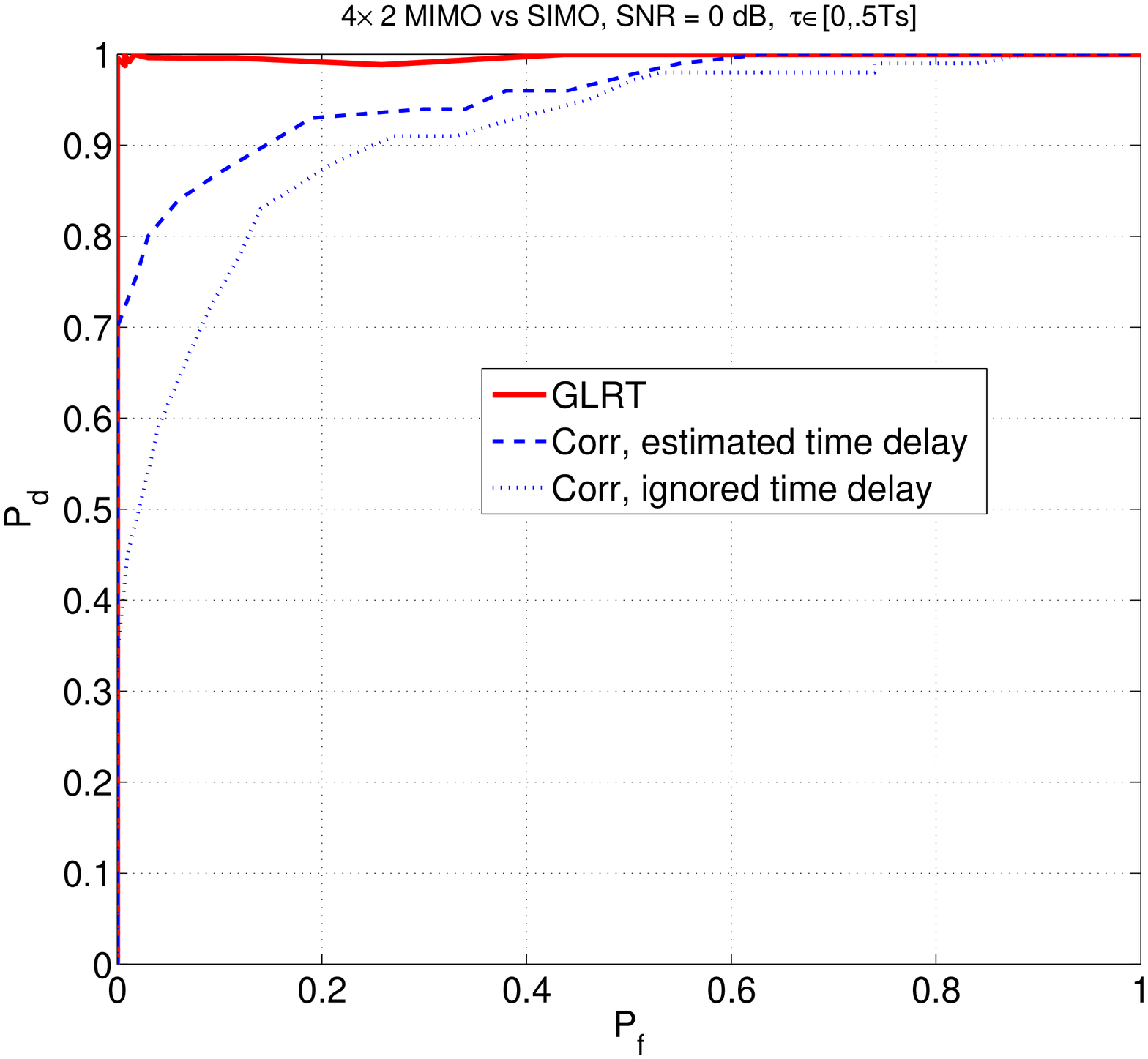}
}
\quad
\subfigure[$4\times 4$ MIMO vs. SIMO, $\mathrm{SNR}=0~dB$]{%
\includegraphics[width=0.4\textwidth,height=!]{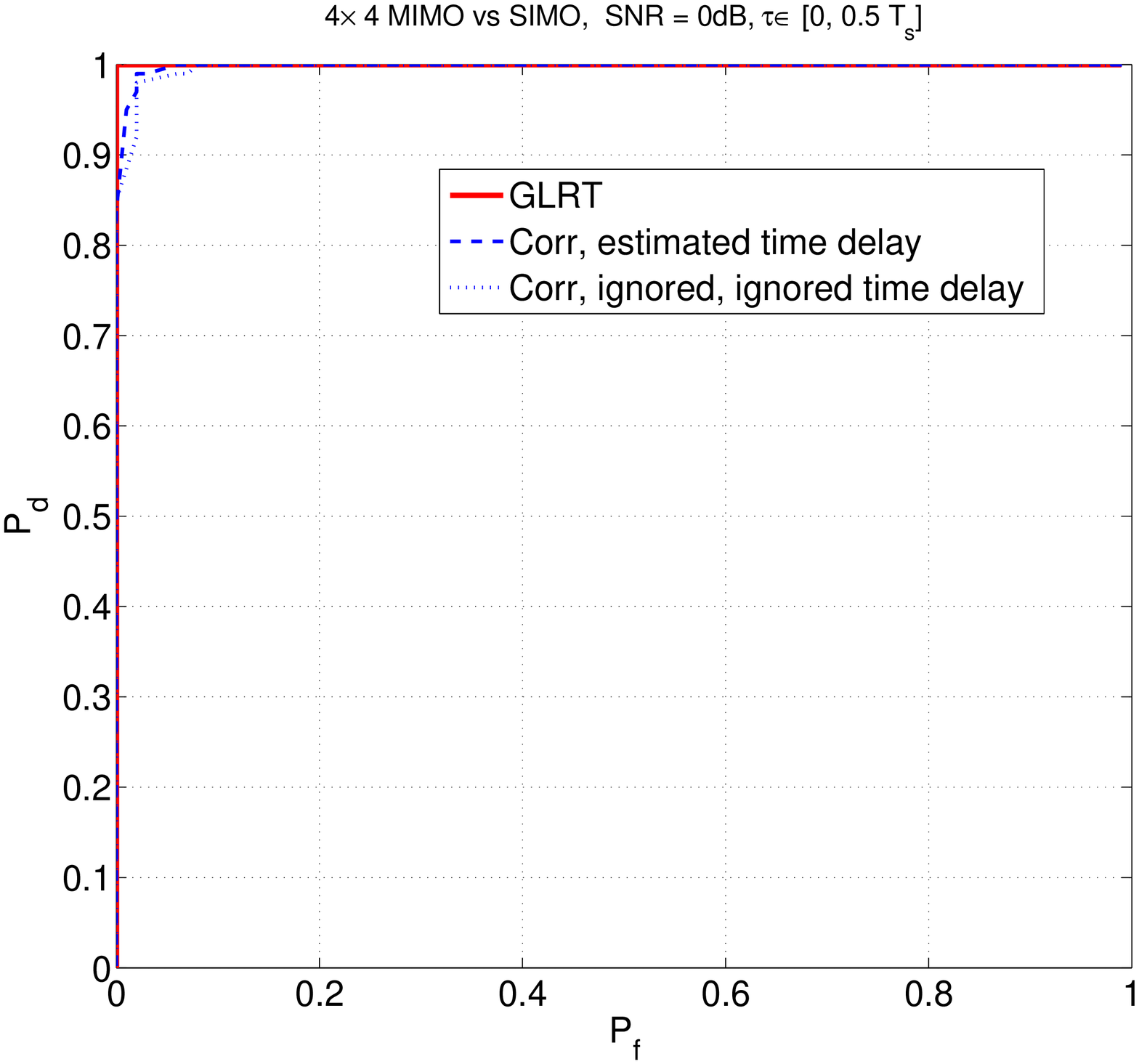}
}
\subfigure[$4\times 2$ MIMO vs. SIMO, $\mathrm{SNR}=-10~dB$]{%
\includegraphics[width=0.4\textwidth,height=!]{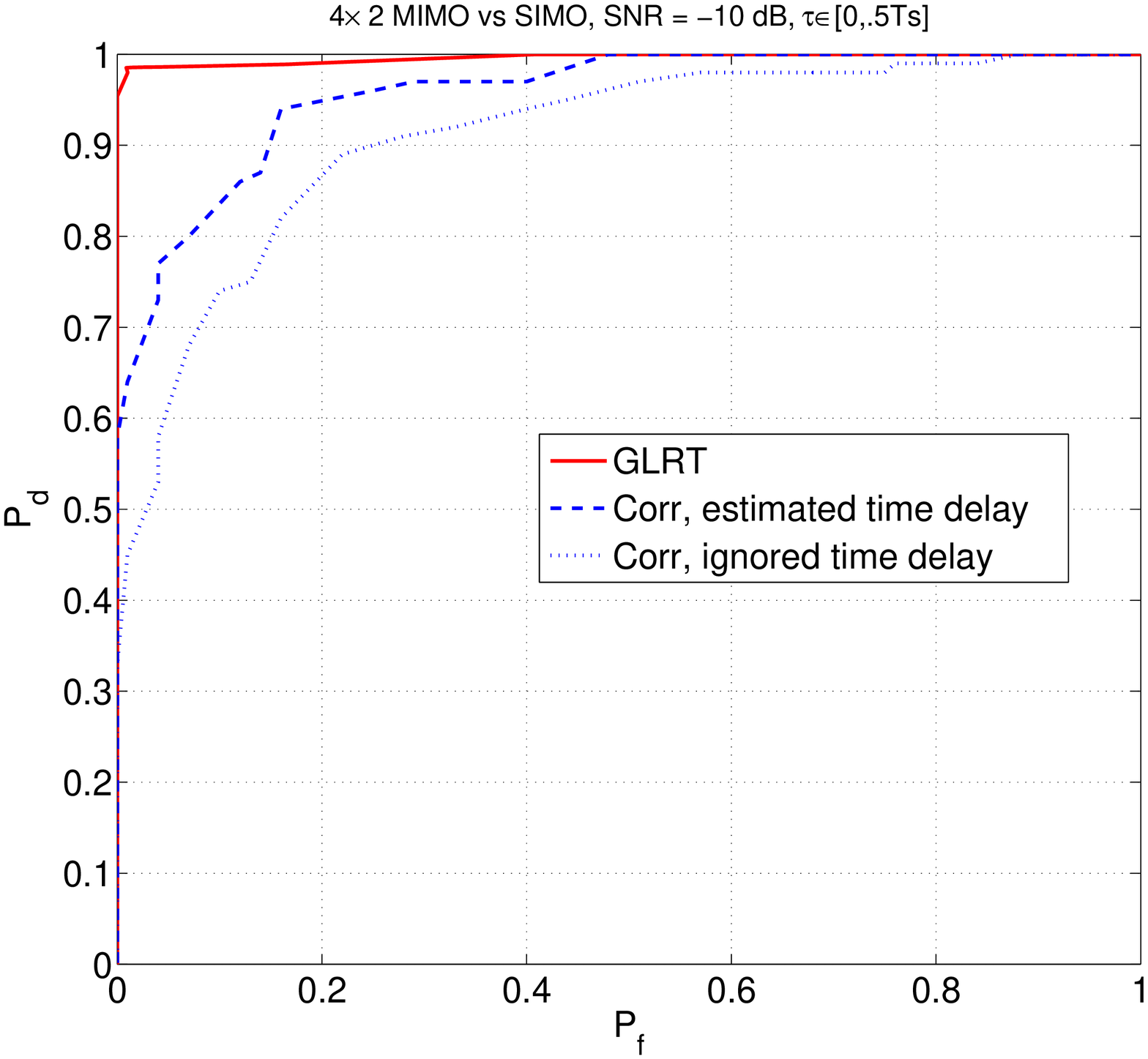}
}
\quad
\subfigure[$4\times 4$ MIMO vs. SIMO, $\mathrm{SNR}=-10~dB$]{%
\includegraphics[width=0.4\textwidth,height=!]{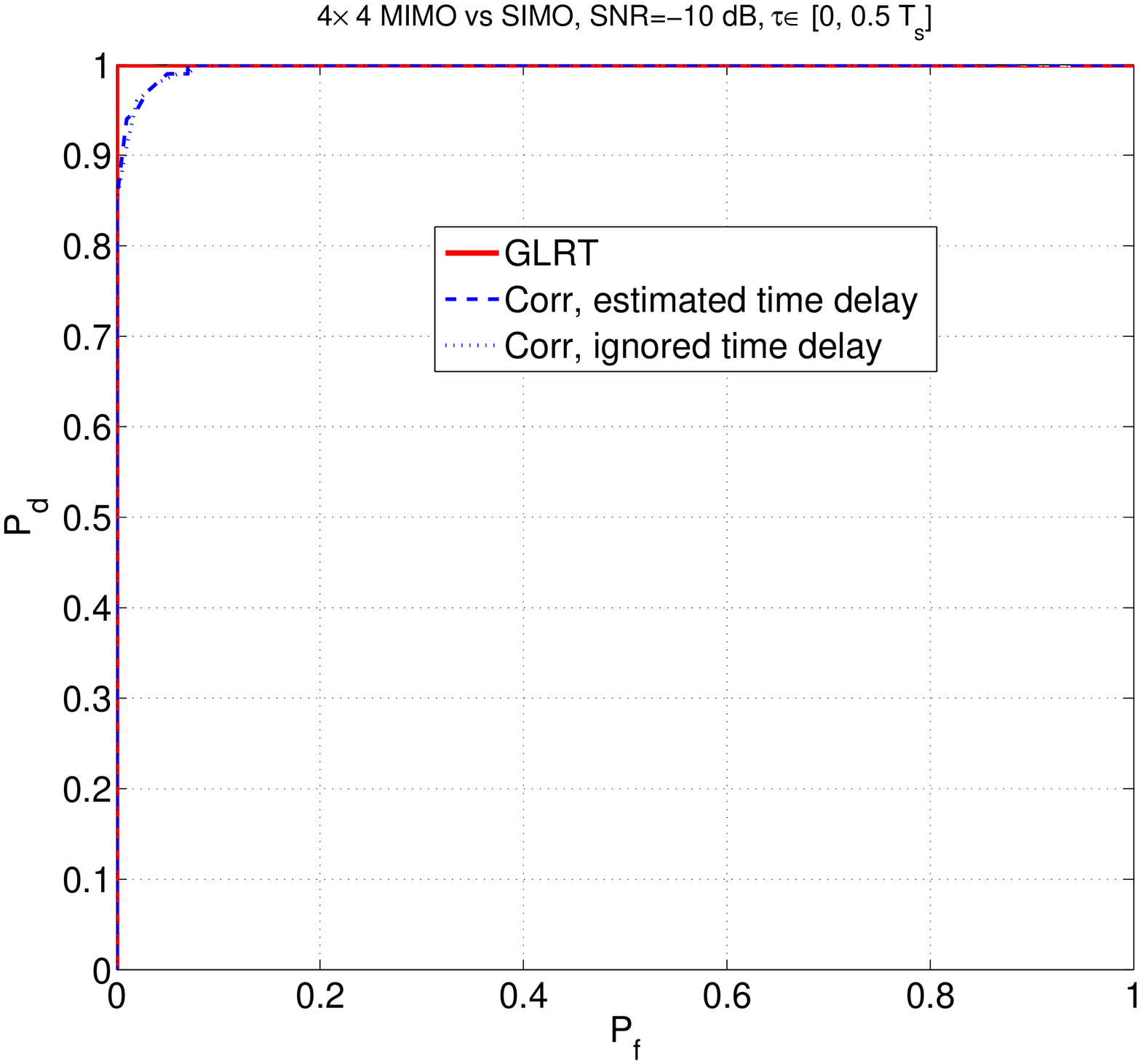}
}
%
\caption{ROC curves for asynchronous MIMO vs. SIMO classification via GLRT and  correlation detectors using  pilot data for symbol timing; $\tau_j$ for $j=1,2$ are selected uniformly randomly from $[0,\frac{1}{2}T_s]$ where $T_s =1$, $L_s=8$}
\label{fig_Pd_snr_timing}
\end{figure*}


\section{Simulation  Results}\label{simulation}
In this section, we provide  numerical results to illustrate the performance of  MIMO system  classification schemes developed in this paper. We assume that there is a maximum of $n=4$ Tx antennas.  The pilot  sequences are chosen to be the columns of $L_s\times L_s$  ($L_p\times L_p$) Hadamard matrix for symbol timing  (channel) estimation.
In all the simulations, the following values for specific parameters are used unless otherwise specified. The transmitted pulse $g(t)$ is assumed to be symmetrically truncated with a roll-off factor $\epsilon =0.3$ and duration $6T_s$. We let $\beta_p=1$, and the pilot sequences with multiple antenna Tx are normalized so that both multiple antenna Tx and single antenna Tx have the same transmit power. The symbol duration $T_s=1$.

\subsection{Performance of Asynchronous  MIMO vs. SIMO Detection }
First, we show the classification performance when the  Tx system is binary; i.e., the problem is to detect  SIMO vs. MIMO.
In Fig. \ref{fig_Pd_snr_timing}, we plot receiver  operating characteristic (ROC) curves to depict the performance of asynchronous GLRT and asynchronous correlation based  MIMO vs. SIMO detection  using  pilot data  used for symbol timing synchronization.  For  the ROC curves, we plot the probability of detection $P_d$ vs. probability of false alarms $P_f$ which are defined below:
 \begin{eqnarray*}
 P_f = Pr (\hat{\mathcal H} = \mathcal H_1| \mathcal H_2), ~~  \mathrm{and} ~~P_d = Pr (\hat{\mathcal H} = \mathcal H_1| \mathcal H_1)
 \end{eqnarray*}
where $\mathcal H_1$ and $\mathcal H_2$ are as defined in (\ref{hyp_MIMO_SIMO_timing}) and  $\hat{\mathcal H}$  denotes  the detection decision. For a given signal-to-noise ratio (SNR), the ROC curves are obtained by varying the threshold and $500$ monte carlo runs were used to generate each curve.  We assume in Fig. \ref{fig_Pd_snr_timing} that the pilot sequence assigned to each Tx antenna is   known at the receiver.  For an  $m\times n$ MIMO system, the first $n$ columns of the Hadamard matrix  are selected.
 In the correlation detector, the delay is estimated as in (\ref{hat_tau_M}) under each  hypothesis.   We further plot the performance of the correlation detector  when the timing delay is ignored; i.e., with the  assumption that  $\hat{\tau}_j=0$ for $j=1,2$. Different sub plots in Fig. \ref{fig_Pd_snr_timing} are for different SNR values and antenna systems. The SNR is defined as $10\log_{10}\left(\frac{||\mathbf H\mathbf R||_F^2}{\mathbb E\{||\mathbf W||_F^2\}}\right)$. As expected, it can be seen that,  the GLRT detector performs better than the correlation detector, especially when the SNR is low and the number of antennas used for MIMO is small. Further, $4\times 4$ MIMO vs. SIMO has a better classification performance  compared to $4\times 2$ MIMO vs. SIMO with both types of detectors.  The correlation detector, which has a simpler implementation compared to GLRT in the asynchronous case,  is more promising  in classifying $4\times 4$ MIMO vs. SIMO compared to classifying $4\times 2$ MIMO vs. SIMO. Further, with $4\times 2$ MIMO, the performance of the correlation detector degrades quite significantly when the time delay is ignored, compared to that with $4\times 4$ MIMO.    When the SNR is changed (e.g  from $0~dB$ to $- 10~dB$), a significant performance degradation is not observed.

  In Fig. \ref{fig_corr_est_ig}, we show the performance of the correlation detector  with ignored time delay  (which has a simpler implementation than that with estimated time delay) as  the  number of Tx antennas for MIMO varies  for a given SNR. We further show the performance with the estimated time delay.
  It can be observed that,  as the number of Tx antennas in MIMO increases, the performance of the correlation detector with the ignored time delay becomes very close to that with the estimated time delay.
\begin{figure}
\centerline{\epsfig{figure=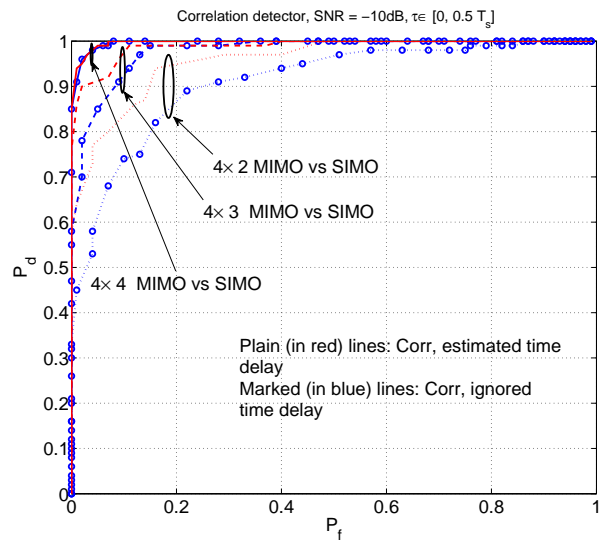,width=9.0cm}}
\caption{ROC curves for asynchronous MIMO vs. SIMO classification via  correlation detector with estimated and ignored time delay; $\tau_j$ for $j=1,2$ are selected uniformly randomly from $[0,\frac{1}{2}T_s]$ where $T_s =1$, $L_s=8$,  $\mathrm{SNR}=-10~dB$}\label{fig_corr_est_ig}
\end{figure}

\begin{figure*}
\centering
\subfigure[$L_s=8$]{%
\includegraphics[width=0.45\textwidth,height=!]{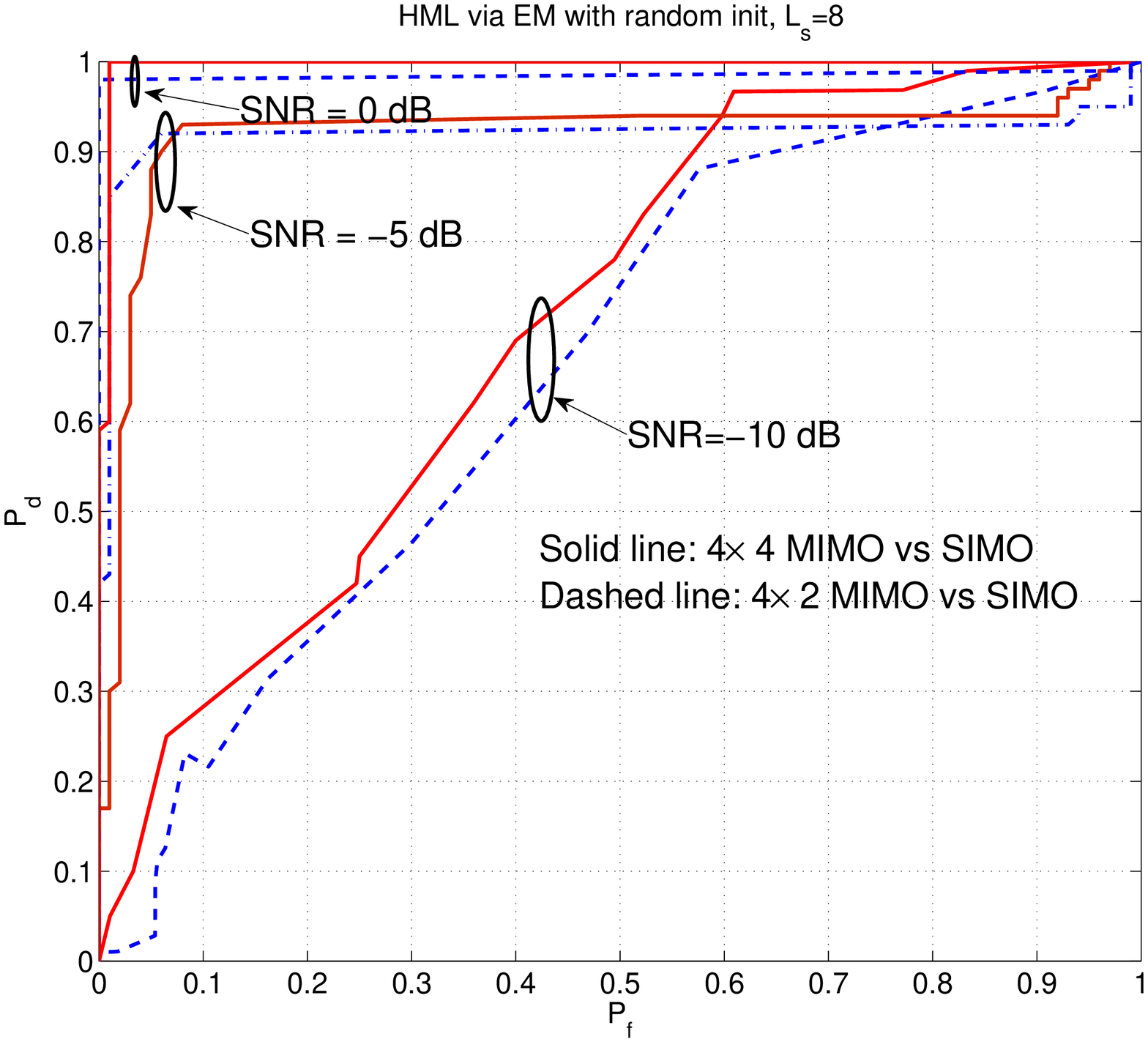}
}
\quad
\subfigure[$L_s=4$]{%
\includegraphics[width=0.45\textwidth,height=!]{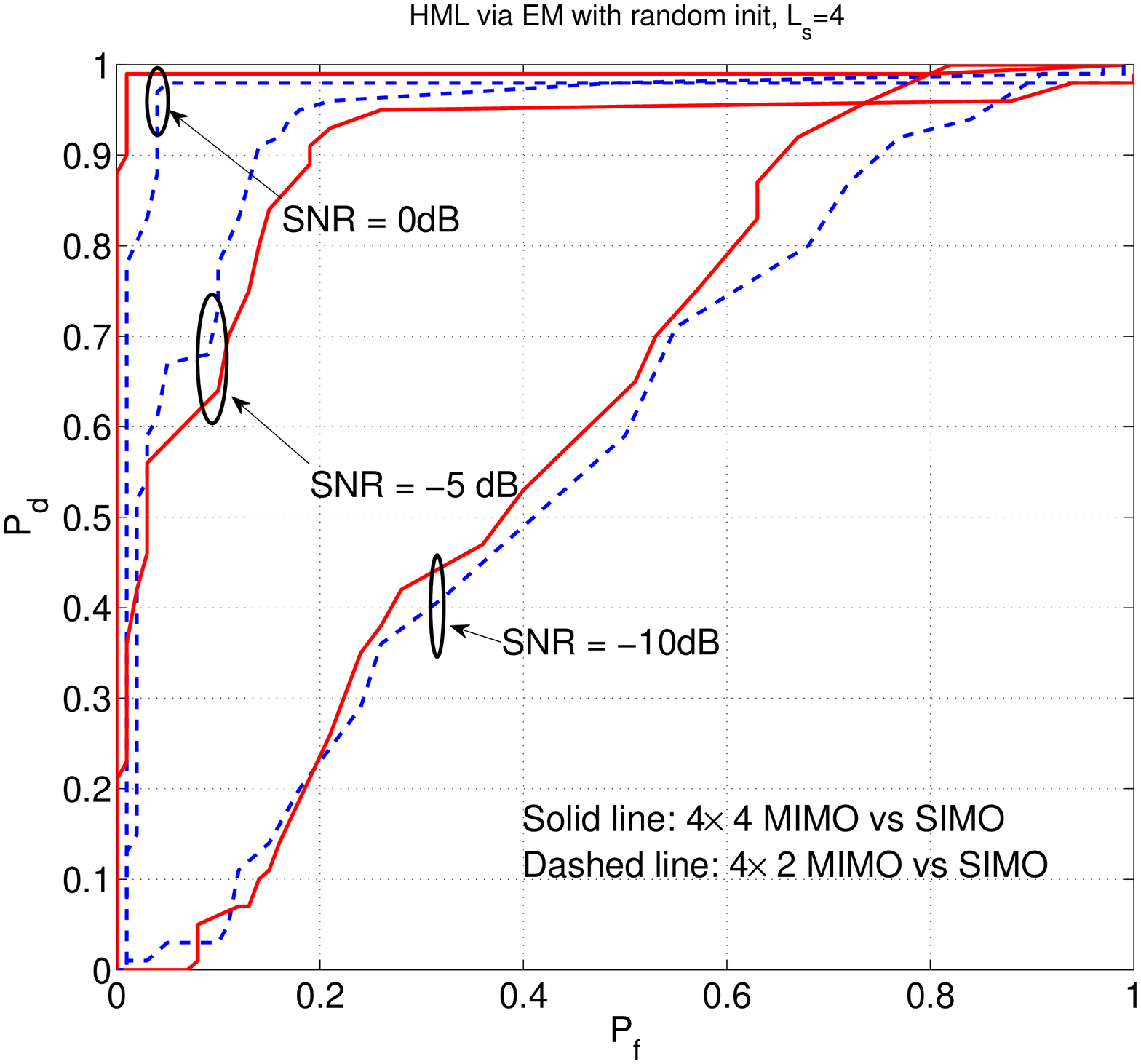}
}

%
\caption{ROC curves for asynchronous MIMO/SIMO classification via HML using  pilot data for symbol timing; $\tau_j$'s for $j=1,2$ are selected uniformly randomly from $[0,\frac{1}{2}T_s]$ where $T_s =1$, random initialization is used for the EM algorithm}
\label{fig_HML_timing}
\end{figure*}

In Fig. \ref{fig_HML_timing}, we illustrate the performance of the HML based classification scheme when the exact knowledge of the pilot sequences is  not available at the receiver. The EM algorithm as discussed in Subsection \ref{HML_EM_Async} is implemented to get the ML estimates of $\mathbf H_j$, and  $\tau_j$
for $j=1,2$. It is well known that the performance of the EM algorithm depends on the initial values of the unknown parameters. In this work, we consider that  the initial values are selected randomly. For $\tau_j$, the initial values are obtained uniformly from $[0, 0.5Ts]$ while for $\mathbf H_j$, the initial values are drawn from a complex normal distribution for $j=1,2$.
  In Fig. \ref{fig_HML_timing}, MIMO vs. SIMO detection  performance with asynchronous HML is plotted in terms of ROC curves. Two subplots correspond to different values for  $L_s$. It can be seen that when the SNR is low,  MIMO vs. SIMO classification has relatively poor performance with the HML algorithm. However, as the SNR increases, the performance of MIMO/SIMO classification improves for both $n=2$ and $n=4$. Further, it is observed that when $L_s$ is small, HML based $4\times 2$ MIMO vs. SIMO performs better (in a small margin) than $4\times 4$ MIMO vs. SIMO for some values of SNR.  It appears that EM algorithm is not capable of finding  good estimates for the  unknown parameters  when $L_s$ is very small. However, as SNR increases, both HML  based $4\times 4$ MIMO vs. SIMO and $4\times 2$ MIMO vs. SIMO detectors  reach almost the perfect detection region.

\subsection{Performance of Synchronous  MIMO vs. SIMO Detection}
In the synchronous case,  only channel matrices are unknown.  In Figures \ref{fig_ROC}-\ref{fig_Pd_snr}, we assume that the exact knowledge of the pilot sequences assigned for each antenna for channel estimation is known at the receiver. In Fig. \ref{fig_ROC}, ROC curves are plotted to illustrate  the performance of the synchronous  GLRT and synchronous correlation  detectors for different $L_p$ and SNR values for MIMO vs. SIMO detection.   As observed in the asynchronous case, it is seen that  GLRT performs better than the correlation detector and this performance gain is more significant as $n$ decreases.   It is noted that in the synchronous case, the correlation detector is pretty simple since it does not require the estimation of any unknown parameters while GLRT estimates the channel matrix.

\begin{figure*}
\centering
\subfigure[$\sigma_w^2=10$, $L_p = 8$, $\mathrm{SNR}=-10~dB$]{%
\includegraphics[width=0.45\textwidth,height=!]{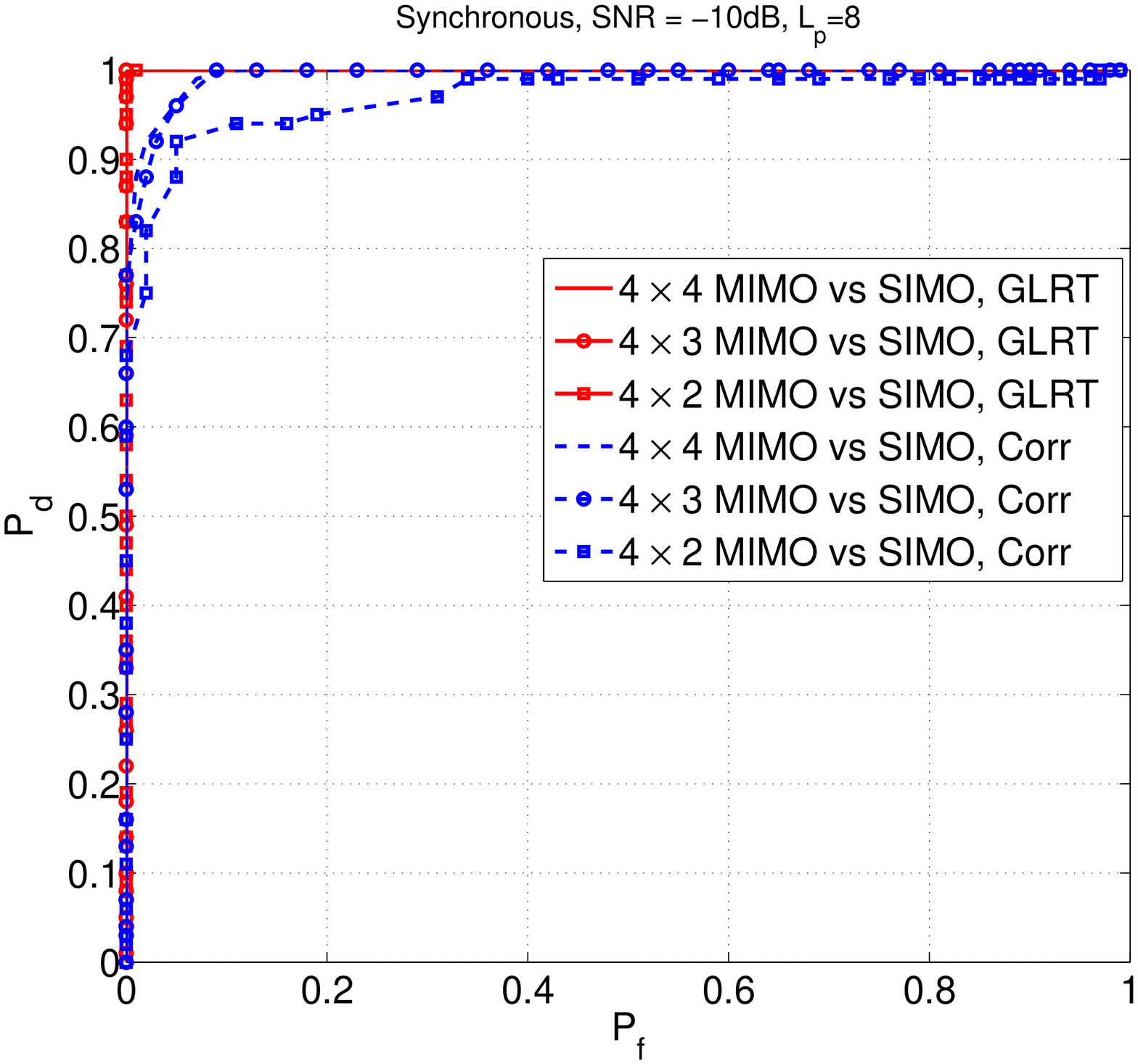}
}
\quad
\subfigure[$\sigma_w^2=10$, $L_p=4$, $\mathrm{SNR}=-10~dB$]{%
\includegraphics[width=0.45\textwidth,height=!]{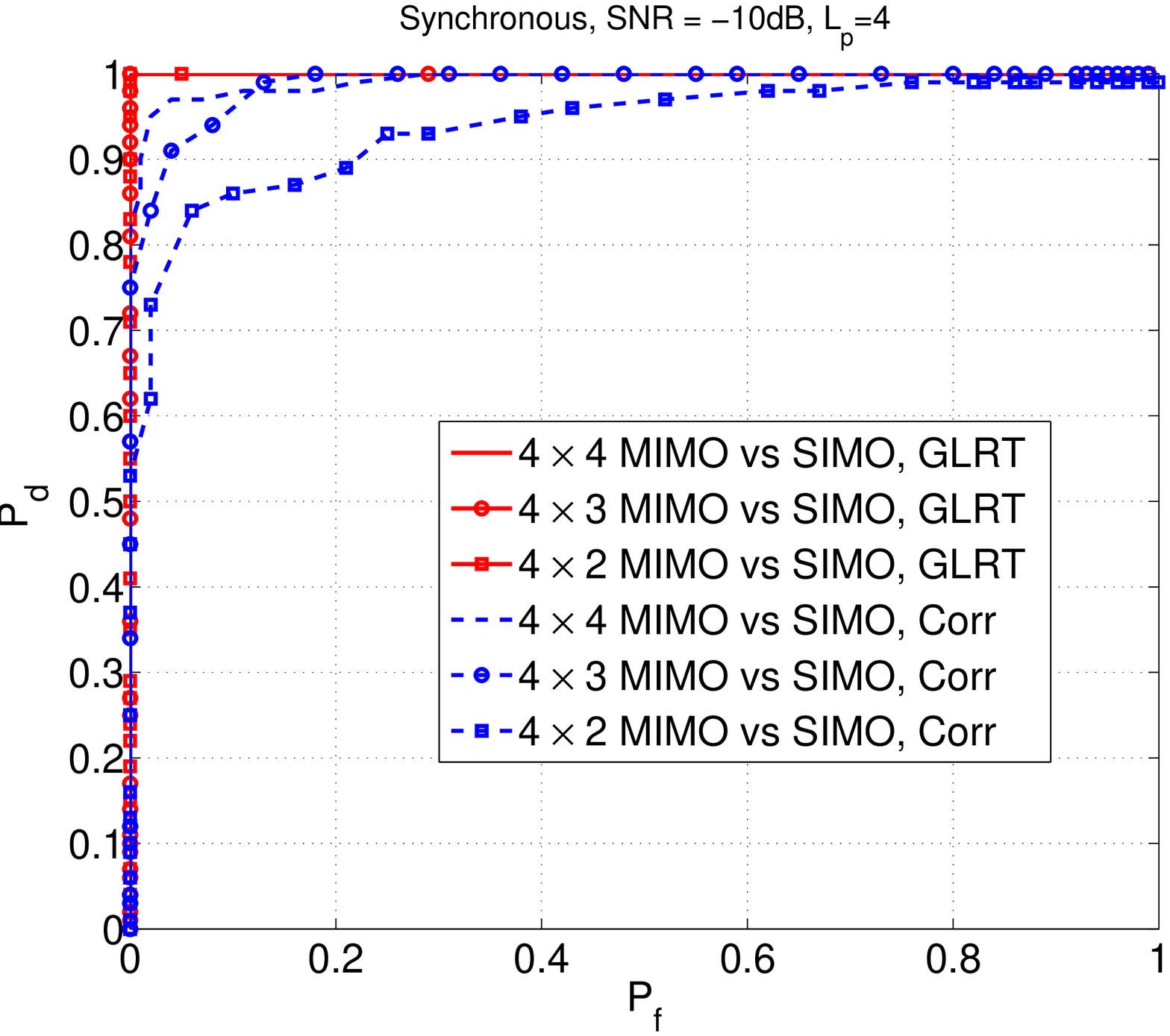}
}
%
\caption{ROC curves for synchronous MIMO/SIMO with  pilot data used  for channel estimation; $T_s=1$}
\label{fig_ROC}
\end{figure*}

In Fig. \ref{fig_sync_async}, we compare the performance of asynchronous and synchronous classification schemes when the exact pilot data is available at the receiver. We consider $4\times 2$ MIMO vs. SIMO classification and for both synchronous and asynchronous cases, we assume that the pilot sequences are of the same length (i.e. $L_s=L_p=8$). Further,  we let $\mathrm{SNR} = -10 ~dB$. It can be seen that, the GLRT detector for both synchronous and asynchronous cases  provides similar performance while the synchronous correlation detector performs  better than the asynchronous correlation detector with a considerable performance gain. It is noted that, the implementation of the asynchronous correlation detector with the estimated time delay is more computationally complex than the synchronous correlation detector. With the asynchronous correlation detector with ignored time delay (which has a similar implementation complexity as that of   the synchronous correlation detector), a significant performance loss can be observed compared to the synchronous correlation detector. Thus, neglecting the unknown time delay leads  to poor MIMO/SIMO classification performance especially  when $n$ is small. However, as seen in Fig. \ref{fig_corr_est_ig}, the  correlation detector with ignored time delay is promising when the  number of Tx antennas is not too small.

\begin{figure}
\centerline{\epsfig{figure=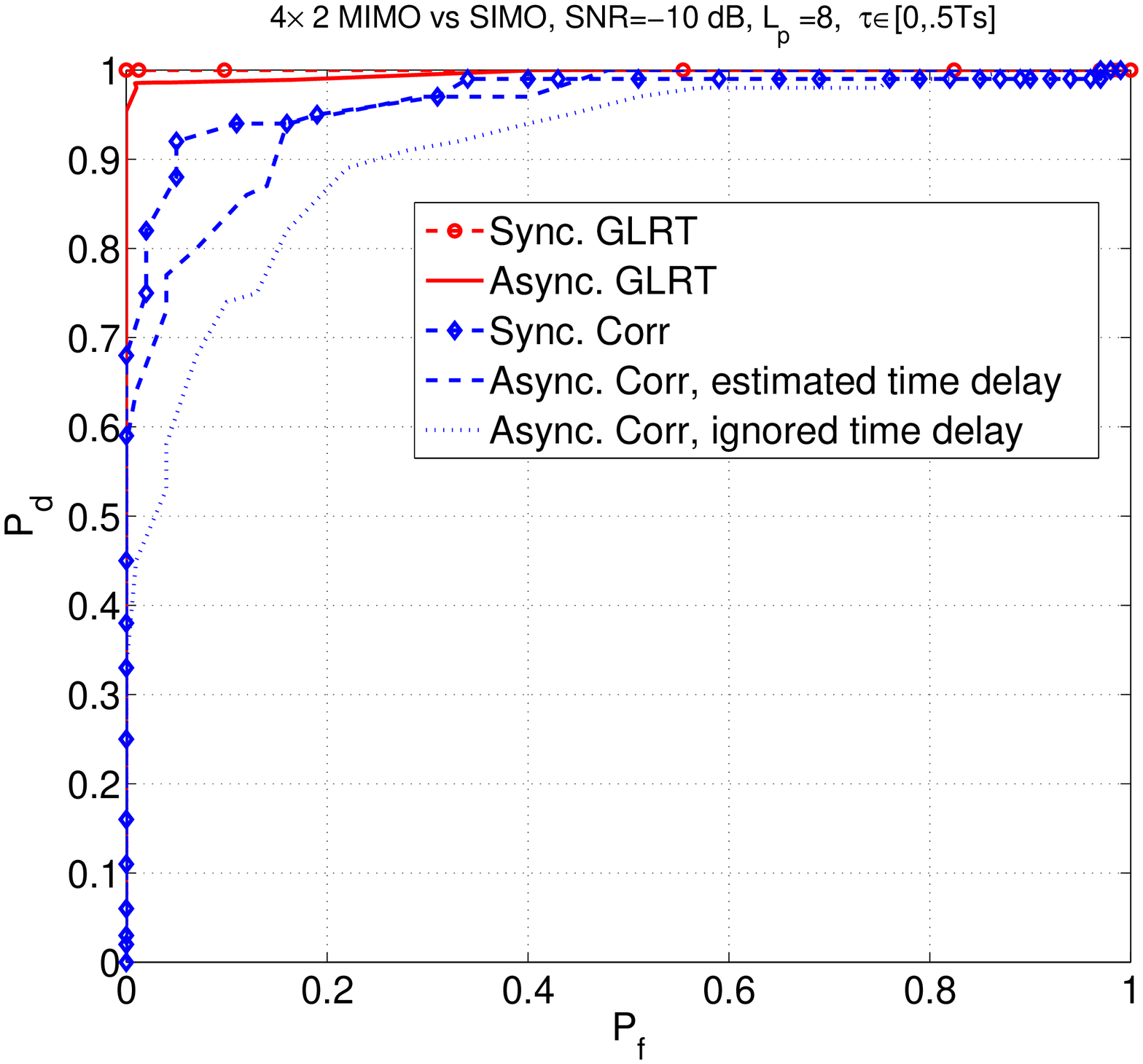,width=9.50cm}}
\caption{Asynchronous and synchronous MIMO/SIMO classification using GLRT  and correlation detectors; $L_p=8$, $\sigma_w^2=10$, $\mathrm{SNR}=-10dB$}\label{fig_sync_async}
\end{figure}

\begin{figure*}
\centering
\subfigure[$\mathrm{SNR}=0~dB$, $L=104$, $L_p=8$, SIMO vs. $4\times 2$ MIMO]{%
\includegraphics[width=0.45\textwidth,height=!]{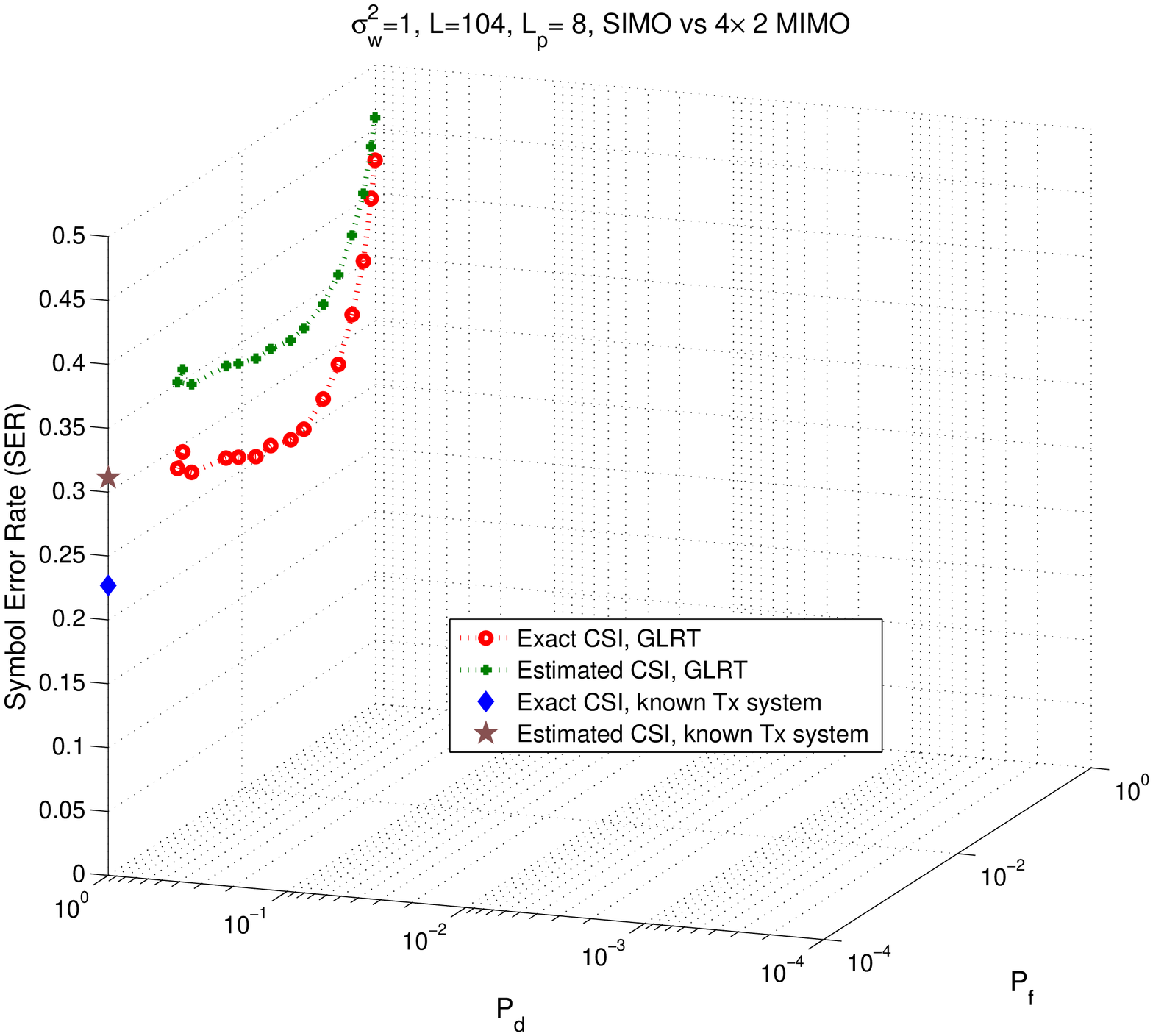}
}
\quad
\subfigure[$\mathrm{SNR}=5~dB$, $L=104$, $L_p=8$, SIMO vs. $4\times 2$ MIMO]{%
\includegraphics[width=0.45\textwidth,height=!]{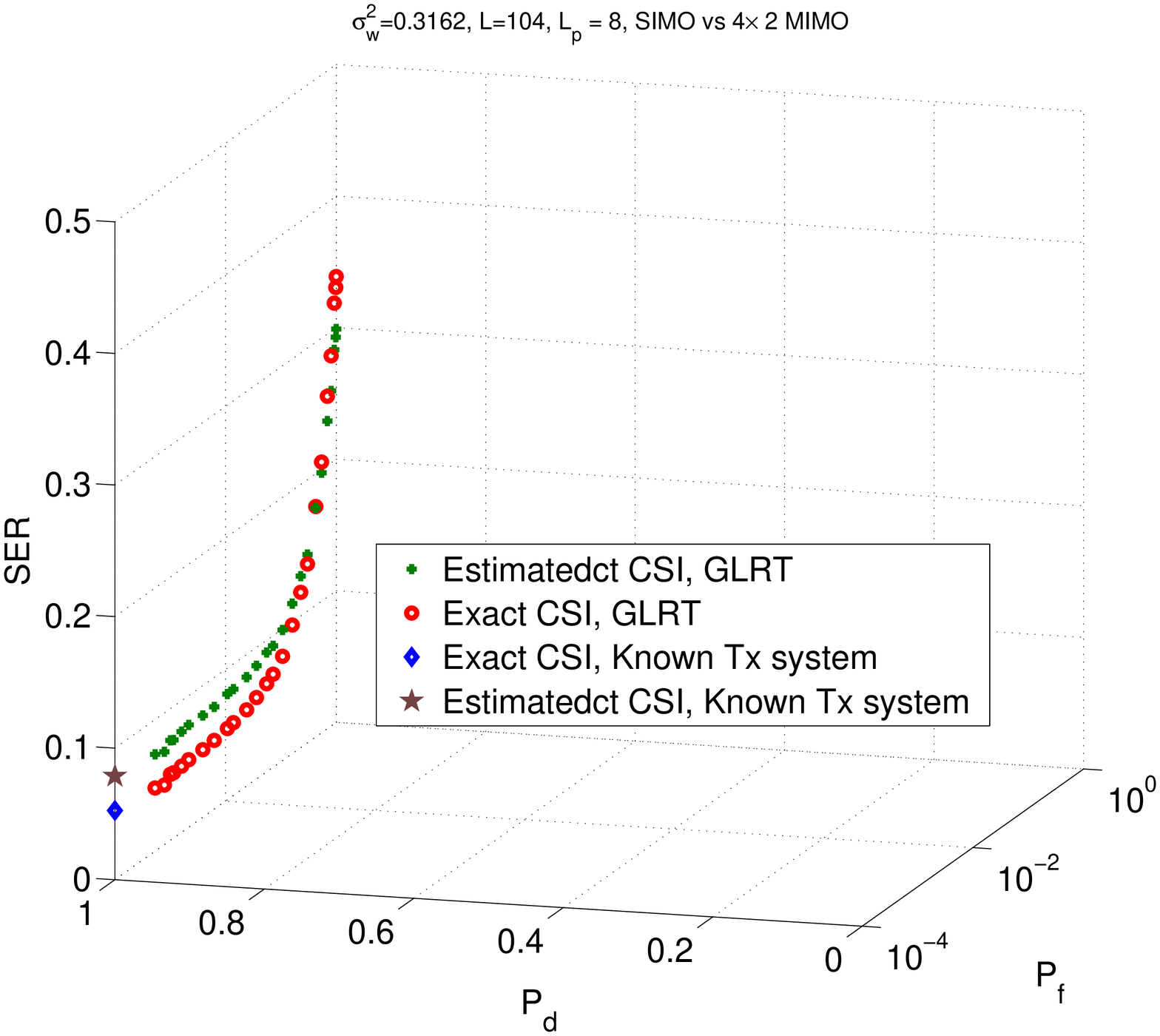}
}
\subfigure[$\mathrm{SNR}=0~dB$, $L=104$, $L_p=8$, SIMO vs. $4\times 4$ MIMO]{%
\includegraphics[width=0.45\textwidth,height=!]{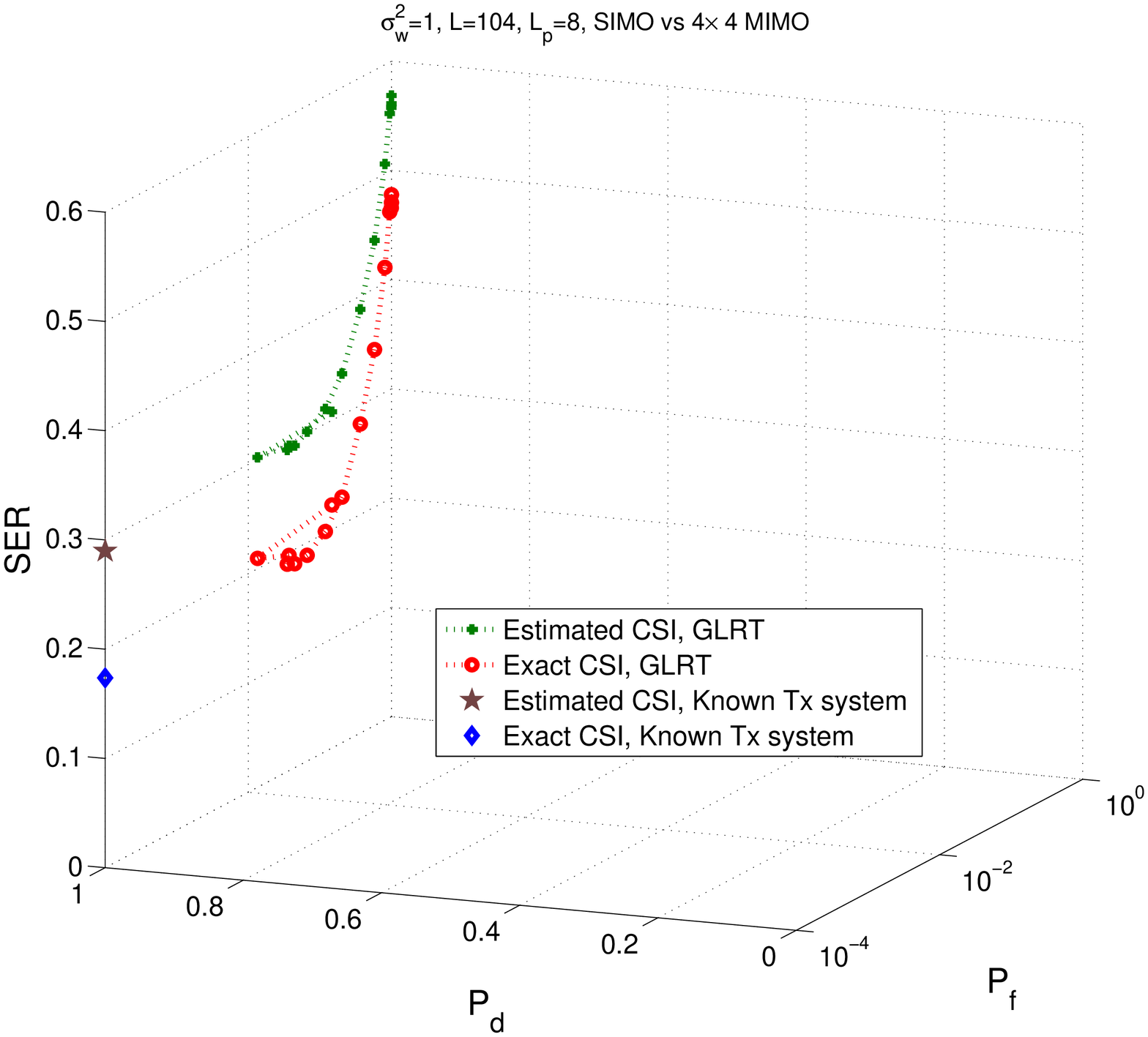}
}
\quad
\subfigure[$\mathrm{SNR}=5~dB$, $L=104$, $L_p=8$, SIMO vs. $4\times 4$ MIMO]{%
\includegraphics[width=0.45\textwidth,height=!]{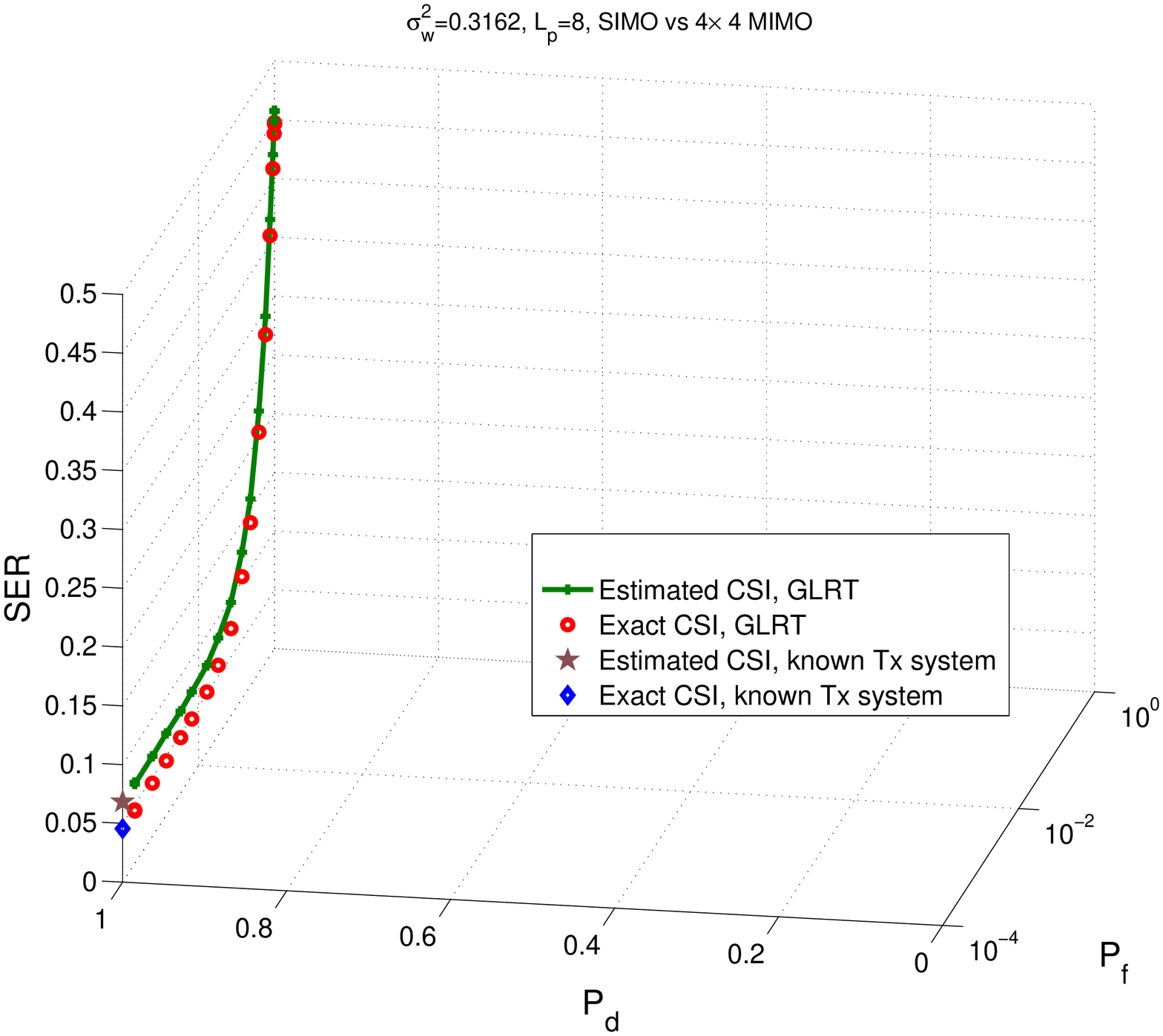}
}
%
\caption{SER when data is recovered via ML after detection  decision is made via GLRT with  pilot data used for channel estimation}
\label{fig_Detection}
\end{figure*}

\begin{figure}
\centerline{\epsfig{figure=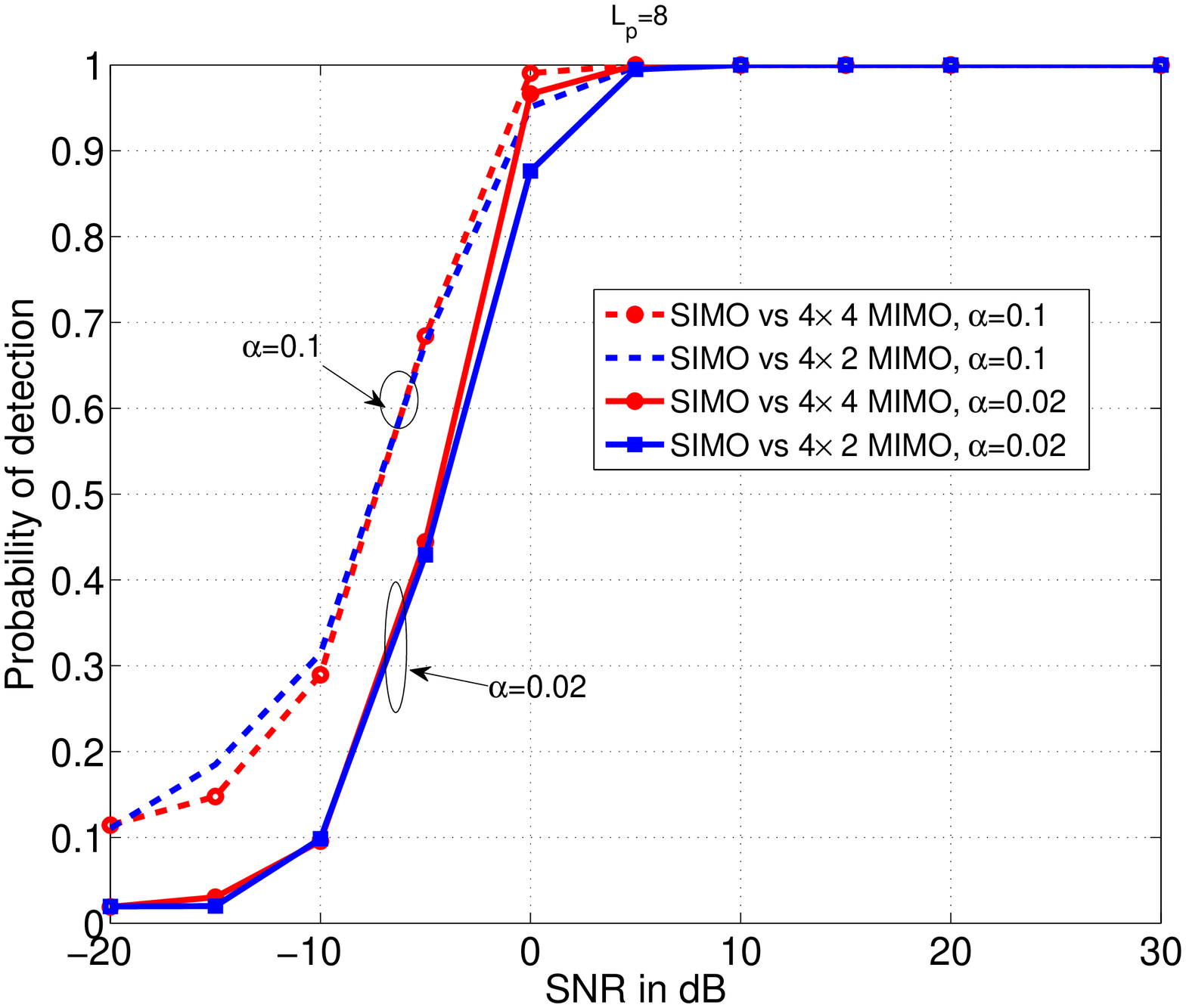,width=9.50cm}}
\caption{Probability of detection (SIMO/MIMO) vs. SNR via  GLRT  in the synchronous case; $L_p=8$, varying $\alpha$}\label{fig_Pd_snr}
\end{figure}

In Fig. \ref{fig_Detection}, we plot the symbol error rate (SER)  when the data is recovered via ML estimation after the detection decision is made via GLRT in the synchronous case. The channel matrix for ML estimation of data  is taken as the one estimated in the GLRT algorithm.  Fig. \ref{fig_Detection}  depicts the recovery performance depending on  the operating point of GLRT; i.e. SER vs. $P_f$ and $P_d$ is plotted.  We further plot the SER when it is exactly known what transmission scheme used (single or multiple). As a reference, we also plot SER with known channel matrix. It can be seen that, when SNR is high (Figs. \ref{fig_Detection} (b) and (d)), SER obtained after making the detection decision based on GLRT with a proper selection of threshold  almost coincides with that when the transmission system is exactly known. However, as the probability of  false alarms increases, the SER also increases. Thus, it is important to design the threshold of the GLRT so that a desired  probability of false alarm is achieved. The design of the threshold for synchronous GLRT in closed-form was  discussed in Section \ref{threshold}.

In Fig. \ref{fig_Pd_snr}, we plot the probability of detection (SIMO/MIMO) vs. SNR for a fixed false alarm rate, $\alpha$ with the synchronous GLRT detector. To achieve a desired value for $\alpha$, the threshold of the GLRT is designed as in (\ref{taug}).   It is observed that when the  average SNR per Rx antenna exceeds $\approx 5~dB$,  perfect detection can be observed at  relatively low values for $\alpha$.
\begin{figure}
\centerline{\epsfig{figure=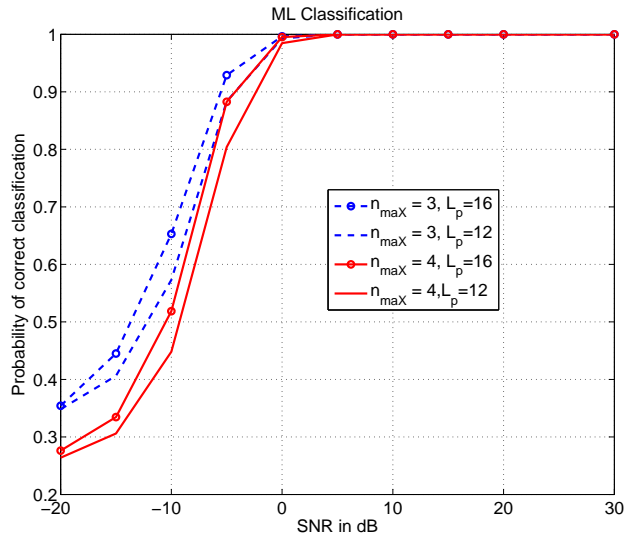,width=9.50cm}}
\caption{Probability of correct Tx system classification  vs. SNR  using ML approach in the synchronous case}\label{fig_Pd_snr_ML}
\end{figure}
\subsection{Performance of ML Classifier}
In Fig. \ref{fig_Pd_snr_ML}, we show  the performance of the ML classifier when the  exact number of antennas used for MIMO is not known. We consider the synchronous case, and assume that  the exact knowledge of the pilot sequences is  available at the receiver. In simulations, the true value for the number of antennas is selected uniformly among $\{1,\cdots, n_{\max}\}$ while two values were considered for $n_{\max}$.  Classification is performed based on the pilot  data used for channel estimation with $L_p=12$ and $L_p=16$. As expected, it can be observed that, the correct classification rate with $3$ Tx antennas is higher than that with when  there are $4$ Tx antennas. As $n_{\max}$ increases, the number of hypotheses also increases leading to weaker distinguishably among hypotheses.    However,  when the SNR exceeds a certain threshold $\approx  0dB$, the ML classification method  provides perfect classification on the Tx system for both values of $n_{\max}$.

\section{Discussion and Summary}\label{conclusion}
 In a typical MIMO system, symbol timing synchronization and channel estimation are required to be performed  before signal demodulation. One common approach to achieve these two tasks is to utilize some pilot data  assigned to each Tx antenna.  In this work, we  have explored how pilot data used for symbol timing synchronization and channel estimation can be exploited to solve  the Tx system classification problem. When perfect timing information is available at the receiver,  Tx system  classification can be performed with the pilot data used for channel estimation. Then, designing ML based classifiers  that estimate unknown parameters  so that the likelihood function is maximized  is not computationally challenging.

 In the case where Tx system is binary (i.e., SIMO vs MIMO with a known number of Tx antennas), it has been shown that, the synchronous GLRT detector with a suitable  threshold classifies  MIMO vs. SIMO perfectly when the SNR  exceeds a certain level. Correlation detector is another simple suboptimal  approach which  shows comparable performance to GLRT in the high SNR region and when the number of antennas used for  MIMO is not very small. However, when the number of Tx antennas in the MIMO system is small,  the correlation detector has a significant  performance   gap compared to  the GLRT detector. When the perfect timing information is not available at the receiver, it also  has to be estimated. In a system where time delay is also estimated via pilot data, we  compute decision statistics for asynchronous  MIMO/SIMO detection  based on that pilot data.   We developed asynchronous GLRT and  correlation detectors  where the estimation of the time delay for both detector requires one-dimensional numerical optimization step.  It was observed that the asynchronous  correlation detector   becomes promising  compared to GLRT  when as the  number of antennas used in the MIMO model increases, even when the time delay is ignored.

In the general multiple hypothesis testing case where the transmission system can take more than two possibilities, it was shown that the performance of the ML classifier degrades as the number of hypotheses increases. However, as the SNR exceeds a certain threshold, perfect classification can be achieved.

In  summary, we discussed several algorithms to classify Tx system  based on pilot data used for timing synchronization and channel estimation. Test statistics for classification were constructed from this  pilot data  without additional overhead. This joint processing is promising in applications where the receiver and the transmitter have limited cooperation and when multiple antennas are used in an adaptive manner to maintain the required quality of service. An interesting future direction is to explore novel techniques to determining the transmission scheme when the transmitter and the receiver are not cooperative at all.

\section*{Appendix A}
\subsection*{Proof of Proposition \ref{prop1}}

Let $\mathbf P_S = U_S U_S^H$ and $\mathbf P_M = U_M U_M^H$ where $U_S= \mathbf R_S^{H}(\mathbf R_S \mathbf R_S^{H})^{-1/2}$ and $U_M=\mathbf R_M^{H}(\mathbf R_M \mathbf R_M^{H})^{-1/2}$ so that $U_S^H U_S = \mathbf I$ and $U_M^H U_M = \mathbf I $. Then we can write,
\begin{eqnarray*}
&~& \mathrm{tr}((\mathbf P_S^{\bot} - \mathbf P_M^{\bot}) \mathbf Y^H \mathbf Y)\nonumber\\
 &=& \sum_{l=1}^m \mathbf y_l^H U_M U_M^H \mathbf y_l - \sum_{l=1}^m \mathbf y_l^H U_S U_S^H \mathbf y_l\nonumber\\
 &=&\sum_{l=1}^m  \mathbf z_l^H \mathbf z_l - \sum_{l=1}^m  \mathbf v_l^H \mathbf v_l=\tilde Z_1 -\tilde Z_2
\end{eqnarray*}
where  $\tilde Z_1 = \sum_{l=1}^m  \mathbf z_l^H \mathbf z_l$ and $\tilde Z_2 = \sum_{l=1}^m  \mathbf v_l^H \mathbf v_l$, $\mathbf y_l$ is the $l$-th row vector of $\mathbf Y$,  $\mathbf z_l = U_M ^H \mathbf y_l$ and  $\mathbf v_l = U_S ^H \mathbf y_l$. We can show that the two random variables $ \frac{2}{\sigma_w^2} \tilde Z_1$ and $\frac{2}{\sigma_w^2}\tilde Z_2$ have Chi-squared distributions under both hypotheses.

It is noted that $ \tilde Z_1$ and $ \tilde Z_2$ are two correlated chi squared random variables and the  computation of the pdf of $\tilde Z =  \tilde Z_1 -  \tilde Z_2$ is difficult. Thus, in the following we approximate $\tilde Z$ to be a Gaussian random variable since using central limit theorem,  we can approximate $\tilde Z_1$ and $\tilde Z_2$ to be Gaussian when $mL_p$ is sufficiently large. Let $ \tilde Z | \mathcal H_2 \sim \mathcal N(\tilde\mu, \tilde\sigma^2)$. It is noted that under $\mathcal H_2$, the $l$-th row vector of $\mathbf Y$ is Gaussian with $\mathbf y_l \sim \mathcal C \mathcal N (\mathbf H_S(l)\tilde{\mathbf R}_S, \sigma_w^2 \mathbf I)$ where $\mathbf H_S(l)$ is the $l$-the element of $\mathbf H_S$ for $l=1,\cdots,m$ and $\tilde{\mathbf R}_S = {\mathbf R}_S^{T}$ is a $L_p\times 1$ vector containing $L_p$ training symbols. Then $\tilde {\mathbf y}_l = \frac{\mathbf y_l}{\sigma_w}  \sim \mathcal C \mathcal N (\frac{\mathbf H_S(l)}{\sigma_w}\tilde{\mathbf R}_S,  \mathbf I)$.
We can compute $\tilde\mu$ as,
\begin{eqnarray*}
\tilde \mu =  (\tilde{\mathbf R}_S^{H} (\mathbf P_M - \mathbf P_S) \tilde{\mathbf R}_S)\sum_{l=1}^m |\mathbf H_S(l)|_2^2  + m~\sigma_w^2 \mathrm{tr}(\mathbf P_M - \mathbf P_S).
\end{eqnarray*}
The variance of  $\tilde Z$, $\tilde\sigma^2$, is given by,
\begin{eqnarray*}
\tilde\sigma^2 = \mathrm{var}(\tilde Z_1) + \mathrm{var}(\tilde Z_2) - 2 \mathrm{cov} (\tilde Z_1,\tilde Z_2).
\end{eqnarray*}
We have,
$
\mathrm{var}(\tilde Z_1) = \sum_{l=1}^m \mathrm{var}\{\mathbf y_l^H \mathbf P_M \mathbf y_l\} $
and
$
\mathrm{var}\{\mathbf y_l^H \mathbf P_M \mathbf y_l\}  = \mathbb E \{(\mathbf y_l^H \mathbf P_M \mathbf y_l)^2\} - (\mathbb E\{\mathbf y_l^H \mathbf P_M \mathbf y_l\})^2.
$
Using results  in (\cite{Bao1}), we can show that,
\begin{eqnarray*}
\mathrm{var}(\tilde Z_1) = 2\sigma_w^2  \tilde{\mathbf R}_S^{H} \mathbf P_M  \tilde{\mathbf R}_S \sum_{l=1}^m{|\mathbf H_S(l)|^2} +  m \sigma_w^4\mathrm{tr}(\mathbf P_M )
\end{eqnarray*}
and similarly,
\begin{eqnarray*}
\mathrm{var}(\tilde Z_2) = 2 \sigma_w^2  \tilde{\mathbf R}_S^{H} \mathbf P_S  \tilde{\mathbf R}_S \sum_{l=1}^m{|\mathbf H_S(l)|^2} +  m \sigma_w^4\mathrm{tr}(\mathbf P_S ).
\end{eqnarray*}
Next we compute  the covariance of $\tilde Z_1$ and $\tilde Z_2$.
We have
\begin{eqnarray}
\mathrm{cov}(\tilde Z_1,\tilde Z_2) = \mathbb E\{\tilde Z_1 \tilde Z_2\} - \mathbb E\{\tilde Z_1 \}\mathbb E\{ \tilde Z_2\}. \label{cov1}
\end{eqnarray}
We can compute $\mathbb E\{\tilde Z_1 \tilde Z_2\}$ as,
\begin{eqnarray}
\mathbb E\{\tilde Z_1 \tilde Z_2\}& =& \mathbb E\left \{  \sum_{l=1}^m \mathbf y_l^H \mathbf P_M \mathbf y_l \sum_{l=1}^m \mathbf y_l^H \mathbf P_S \mathbf y_l  \right\}\nonumber\\
&=&  \sigma_w^4\sum_{l=1}^m  \mathbb E\left \{  \tilde{\mathbf y}_l^H \mathbf P_M \tilde{\mathbf y}_l \tilde{\mathbf y}_l^H \mathbf P_S \tilde{\mathbf y}  \right\}  \nonumber\\
&+& \sigma_w^4\underset{l\neq j}{\sum}\mathbb E\{\tilde{\mathbf y}_l^H \mathbf P_M\tilde{\mathbf y}_l\}\mathbb E\{ \tilde{\mathbf y}_j^H \mathbf P_S \tilde{\mathbf y}_j \}. \label{cov2}
\end{eqnarray}
We have,
\begin{eqnarray*}
\mathrm{cov}(\tilde Z_1,\tilde Z_2) &=&  \sigma_w^4 \sum_{l=1}^m  \mathbb E\left \{  \tilde{\mathbf y}_l^H \mathbf P_M \tilde{\mathbf y}_l \tilde{\mathbf y}_l^H \mathbf P_S \tilde{\mathbf y}  \right\}   \nonumber\\
&-&\sigma_w^4 {\sum}_{l=1}^m \mathbb E\{\tilde{\mathbf y}_l^H \mathbf P_M\tilde{\mathbf y}_l\}\mathbb E\{ \tilde{\mathbf y}_l^H \mathbf P_S \tilde{\mathbf y}_l \}.
\end{eqnarray*}
It can be shown that,
\begin{eqnarray*}
\mathrm{cov}(\tilde Z_1,\tilde Z_2)   &=& 2 \sigma_w^2  \tilde{\mathbf R}_S^{H} \mathbf P_M \mathbf P_S  \tilde{\mathbf R}_S \sum_{l=1}^m|\mathbf H_S(l)|^2 \nonumber\\
&+&  m \sigma_w^4 \mathrm{tr}(\mathbf P_M \mathbf P_S).
\end{eqnarray*}
Then the variance is given by,
 \begin{eqnarray}
 \tilde\sigma^2 &=& 2 \sigma_w^2 ( \tilde{\mathbf R}_S^{H}( \mathbf P_M  + \mathbf P_S - 2\mathbf P_M \mathbf P_S) \tilde{\mathbf R}_S \sum_{l=1}^m{|\mathbf H_S(l)|^2} \nonumber\\
 &+&  m \sigma_w^4\mathrm{tr}(\mathbf P_M+\mathbf P_S -2\mathbf P_M \mathbf P_S). \label{tsigma}
 \end{eqnarray}
Since $\mathbf H_S$ is to be estimated, the threshold is found based on the estimated values for $\mathbf H_S$  and approximate values for $\tilde\mu$ and $\tilde\sigma$ are obtained accordingly.
Thus, we have,
\begin{eqnarray*}
 Pr (\Lambda_{GLRT,sync} \geq \tau_g |\mathcal H_2 ) &\approx & Pr (\tilde Z \geq  \tau_g)\approx  Q\left(\frac{{\tau_g}{} - \tilde \mu}{\tilde\sigma}\right)
\end{eqnarray*}
completing the proof.

\bibliographystyle{IEEEtran}
\bibliography{IEEEabrv,bib1}

\end{document}